\documentclass[11pt, letterpaper,USenglish]{scrartcl}
\usepackage[utf8]{inputenc}
\usepackage[T1]{fontenc}

\usepackage{amsfonts}
\usepackage{amsmath}
\usepackage{caption}
\usepackage{comment}
\usepackage{environ}
\usepackage{scrlayer-scrpage}
\usepackage{float}
\usepackage{fontaxes}
\usepackage{graphics}
\usepackage{hyperref}
\usepackage{inconsolata}
\usepackage{libertine}
\usepackage{manyfoot}
\usepackage{microtype}
\usepackage{mweights}
\usepackage{nccfoots}
\usepackage{setspace}
\usepackage{totpages}
\usepackage{trimspaces}
\usepackage{upquote}
\usepackage{url}
\usepackage{xcolor}
\usepackage{xkeyval}
\usepackage{mathtools}
\usepackage{amssymb}

\usepackage{tikz}

\usepackage[left=1in,right=1in,top=1in,bottom=1in]{geometry}

\definecolor{dark-blue}{rgb}{0.05,0.25,0.85}
\hypersetup{colorlinks=true,linkcolor=dark-blue,citecolor=dark-blue}

\makeatletter
\renewcommand\footnotesize{%
   \@setfontsize\footnotesize\@ixpt{11}%
   \abovedisplayskip 8\p@ \@plus2\p@ \@minus4\p@
   \abovedisplayshortskip \z@ \@plus\p@
   \belowdisplayshortskip 4\p@ \@plus2\p@ \@minus2\p@
   \def\@listi{\leftmargin\leftmargini
               \topsep 4\p@ \@plus2\p@ \@minus2\p@
               \parsep 2\p@ \@plus\p@ \@minus\p@
               \itemsep \parsep}%
   \belowdisplayskip \abovedisplayskip
}
\makeatother

\usepackage[amsmath,thmmarks,hyperref]{ntheorem}
\usepackage{cleveref}

\crefformat{page}{#2page~#1#3}%
\Crefformat{page}{#2Page~#1#3}%
\crefformat{equation}{#2(#1)#3}%
\Crefformat{equation}{#2(#1)#3}%
\crefformat{figure}{#2Figure~#1#3}%
\Crefformat{figure}{#2Figure~#1#3}%
\crefformat{section}{#2Section~#1#3}
\Crefformat{section}{#2Section~#1#3}
\crefformat{chapter}{#2Chapter~#1#3}
\Crefformat{chapter}{#2Chapter~#1#3}
\crefformat{chapter*}{#2Chapter~#1#3}
\Crefformat{chapter*}{#2Chapter~#1#3}
\crefformat{part}{#2Part~#1#3}
\Crefformat{part}{#2Part~#1#3}
\crefformat{enumi}{#2(#1)#3}
\Crefformat{enumi}{#2(#1)#3}

\theoremnumbering{arabic}
\theoremstyle{plain}
\theoremsymbol{}
\theorembodyfont{\itshape}
\theoremheaderfont{\normalfont\bfseries}
\theoremseparator{}

\newtheorem{theorem}{Theorem}[section]
\crefformat{theorem}{#2Theorem~#1#3}
\Crefformat{theorem}{#2Theorem~#1#3}

\newcommand{\newtheoremwithcrefformat}[2]{%
  \newtheorem{#1}[theorem]{#2}%
  \crefformat{#1}{##2\MakeUppercase#1~##1##3}%
  \Crefformat{#1}{##2\MakeUppercase#1~##1##3}%
}

\newtheoremwithcrefformat{proposition}{Proposition}
\newtheoremwithcrefformat{observation}{Observation}
\newtheoremwithcrefformat{lemma}{Lemma}
\newtheoremwithcrefformat{conjecture}{Conjecture}
\newtheoremwithcrefformat{corollary}{Corollary}
\theorembodyfont{\upshape}
\newtheoremwithcrefformat{example}{Example}
\newtheoremwithcrefformat{remark}{Remark}
\newtheoremwithcrefformat{definition}{Definition}
\newtheoremwithcrefformat{claim}{Claim}

\theoremsymbol{\ensuremath{\square}}

\theoremstyle{nonumberplain}
\theoremheaderfont{\scshape}
\theorembodyfont{\normalfont}
\theoremsymbol{\ensuremath{\square}}
\newtheorem{proof}{Proof.}
\newcommand{\sth}{\mathrel : }
\newcommand{\bty}{{\bar T_Y}}
\newcommand{\Sh}{S_{>d}} 
\newcommand{\Sl}{S_{\leq d}} 
\newcommand{\Th}{T_{>d}} 
\newcommand{\Tl}{T_{\leq d}} 
\newcommand{\DST}{\textsc{Dst}}
\newcommand{\DSTG}{\textsc{Dst}(G, r, T, k)}

\newenvironment{cenv}{\begin{list}{}{%
      \setlength{\labelwidth}{1.5em}%
      \setlength{\leftmargin}{\labelwidth}%
      \addtolength{\leftmargin}{\labelsep}%
      \setlength{\listparindent}{0em}%
      \setlength{\topsep}{10pt}%
      \setlength{\itemsep}{5pt}%
      \setlength{\parsep}{0pt}%
    }
  }{
  \end{list}
}

\newcounter{claimcounter}

\newenvironment{Claim}{
  
  \refstepcounter{claimcounter}
  \begin{cenv}
  \item[{Claim \arabic{claimcounter}.}]
  }{
  \end{cenv}
}
\newenvironment{ClaimProof}[1][]{\noindent{%
\ifthenelse{\equal{#1}{}}{{\itshape Proof.\ }}{{\itshape #1.\ }}%
}}{\hspace*{1em}\nobreak\hfill$\dashv$\endtrivlist\addvspace{2ex plus
0.5ex minus0.1ex}}

\newcommand{\arrws}{\tikz {\draw[->] (0pt,0pt)--(4pt,0pt); \draw[<-] (0pt,-3pt)--(4pt,-3pt);}}
\newcommand{\wcol}{\mathrm{wcol}}
\newcommand{\dwcol}[1]{\mathrm{wcol}^{\arrws}_{#1}}

\newcommand{\dadm}[1]{\mathrm{adm}^{\arrws}_{#1}}
\newcommand{\dWReach}[1]{\mathrm{WReach}^{\arrws}_{#1}}

\newcommand{\col}{\mathrm{col}}

\newcommand{\Oof}{\mathcal{O}}
\newcommand{\CCC}{\mathcal{C}}
\newcommand{\FFF}{\mathcal{F}}
\newcommand{\PPP}{\mathcal{P}}

\newcommand{\minor}{\preccurlyeq}
\newcommand{\N}{\mathbb{N}}
\newcommand{\dist}{\mathrm{dist}}

\newcommand{\td}{\mathrm{td}}

\def\grad_#1{\nabla\!_{#1}}

\def\cqedsymbol{\ifmmode$\lrcorner$\else{\unskip\nobreak\hfil
\penalty50\hskip1em\null\nobreak\hfil$\lrcorner$
\parfillskip=0pt\finalhyphendemerits=0\endgraf}\fi}

\renewcommand{\mid}{~:~}
\newcommand{\ie}{i.e.\@ }
\newcommand{\abs}[1]{\ensuremath{\left\lvert#1\right\rvert}}

\newcommand{\tcb}[1]{\textcolor{blue}{#1}}
\usepackage[normalem]{ulem}

\begin{document}

\thispagestyle{empty}
\vspace{-0.2cm}
\begin{center}
  \huge{\textbf {Algorithmic Properties of Sparse Digraphs}}\\[0.6cm]

  \Large Stephan Kreutzer

  \large Technische Universit\"at Berlin, Germany
  
  \normalsize \texttt{stephan.kreutzer@tu-berlin.de}
  
\vspace{0.3cm}

  \Large Patrice Ossona de Mendez
  
  \large Centre d'Analyse et de Math\'ematiques Sociales (CNRS, UMR 8557), Paris, France
  
  --- and ---
 
 Computer Science Institute of Charles University (IUUK), Prague, Czech Republic
  
  \normalsize \texttt{pom@ehess.fr}
  
  \vspace{0.3cm}
  
  \parbox{8cm}{\centering \Large Roman Rabinovich
  
  \large Technische Universit\"at Berlin, Germany
  
  \normalsize \texttt{roman.rabinovich@tu-berlin.de}}
 % 
%  \vspace{0.3cm}
%
 \parbox{6cm}{ \centering 
  \Large Sebastian Siebertz

  \large University of Warsaw, Poland
  
  \normalsize \texttt{siebertz@mimuw.edu.pl}}
  
  \vspace{0.4cm}    
\end{center}

\begin{abstract}
\noindent The notions of bounded expansion~\cite{nevsetvril2008grad} 
and nowhere 
denseness~\cite{nevsetvril2011nowhere}, introduced by Ne\v{s}et\v{r}il and 
Ossona de Mendez as structural measures for undirected graphs,
have been applied very successfully in algorithmic graph theory. We study the 
corresponding notions of directed bounded expansion and nowhere crownfulness
on directed graphs, as introduced by Kreutzer and Tazari in~\cite{kreutzer2012directed}. 
These classes are very general classes of sparse directed graphs, as they include, on 
one hand, all classes of directed graphs whose underlying undirected class has
bounded expansion, such as planar, bounded-genus, and $H$-minor-free graphs, 
and on the other hand, they also contain classes whose underlying class is not 
nowhere dense. 

  We show that many of the algorithmic tools that were developed for undirected bounded expansion classes can, with some care, also be applied in their directed counterparts, and thereby we highlight a rich algorithmic structure theory of directed bounded expansion classes.
  
%  classes of digraphs of directed bounded expansion have
%  very similar properties to their undirected counterparts and many of the algorithmic
%tools that where developed for undirected bounded expansion can be
%proved also in the directed setting.

More specifically, we show that the directed Steiner tree problem is fixed-parameter tractable 
on any class of directed bounded expansion parameterized by the number $k$ of non-terminals
plus the maximal diameter $s$ of a strongly connected component in the
subgraph induced by the terminals. 
Our result strongly generalizes a result of Jones et al.~\cite{Jones13}, who
proved that the problem is fixed parameter tractable on digraphs 
of bounded degeneracy if the set of terminals is required to be acyclic. 

We furthermore prove that for every integer $r\geq 1$, the distance-$r$ 
dominating set problem can be approximated up to a factor $\Oof(\log k)$ 
and the connected distance-$r$ dominating set problem can be approximated
up to a factor $\Oof(k\cdot \log k)$ on any class of directed bounded 
expansion, where $k$ denotes the size of an optimal solution. If furthermore, 
the class is nowhere crownful, we are able to compute a polynomial kernel
for distance-$r$ dominating sets. Polynomial kernels for this problem were 
not known to exist on any other existing digraph measure for sparse classes.
\end{abstract}

\hrulefill

\scriptsize

\noindent The research of Stephan Kreutzer and Roman Rabinovich 
is supported by the
European Research Council (ERC) under the European Union's Horizon
2020 research and innovation program (ERC consolidator grant DISTRUCT,
agreement No.\ 648527).
The research of Patrice Ossona de Mendez is supported by grant ERCCZ LL-1201 and by the European Associated Laboratory ``Structures in
Combinatorics'' (LEA STRUCO), and partially supported by ANR project Stint under reference ANR-13-BS02-0007.
The research of Sebastian Siebertz is supported by the National Science Centre of Poland via POLONEZ grant agreement UMO-2015/19/P/ST6/03998, 
which has received funding from the European Union's Horizon 2020 research and 
innovation program (Marie Sk\l odowska-Curie grant agreement No.\ 665778).

\normalsize
\setcounter{page}{0}
\pagebreak
%%% Local Variables:
%%% mode: latex
%%% TeX-master: "focs"
%%% End:

\section{Introduction}

  Structural graph theory has made a
  deep impact on the analysis of complex algorithmic graph problems
  and the design of graph algorithms for hard problems. It provides a
  wealth of new and different tools for dealing with the intrinsic
  complexity of NP-hard problems on graphs and these methods have been
  applied very successfully  in algorithmic graph theory, in approximation, optimization or  the design of
  exact and parameterized algorithms for problems on
  undirected 
graphs, see e.g.~\cite{Chuzhoy16,DemaineHaj08,demaine2005subexponential,
DemaineHajKaw05,
DemaineH04,DemaineH05,
FominLRS11,FominLST10,GMXIII}

  Concepts such as \emph{tree width} or \emph{excluded (topological)
    minors} as well as
  density based graph parameters such as \emph{bounded expansion}~\cite{nevsetvril2008grad} or \emph{nowhere
  denseness}~\cite{nevsetvril2011nowhere} capture important properties of graphs and make them
  applicable for algorithmic applications. 

  Developing a similar structural theory for directed graphs resulting
  in classes of digraphs with a similarly broad 
  algorithmic impact has so far not been blue as successful as
  for undirected graphs. The general goal
  is to identify structural parameters of digraphs which define
  interesting and general classes of digraphs for which at the same
  time we have a rich set of algorithmic tools available that can be used in the
  design of algorithms on these classes. However,  essentially all
  approaches, e.g.~in~\cite{barat2006directed,berwanger2006dag,ganian2009digraph,hunter2008digraph,obdrvzalek2006dag,safari2005d}, of generalizing even the well-understood and fairly basic
  concept of tree width to digraphs have failed to produce digraph parameters
  that come even near the wide spectrum of algorithmic applications
  that tree width has found. This has even led to claims that this
  program cannot be successful and that such measures for digraphs
  cannot exist \cite{GanianHK0ORS16}. 

  The main conceptual contribution of this paper is to finally give a
  positive example of a digraph parameter that
  satisfies the conditions of the program outlined above: we identify a very general type
  of digraph classes which have a similar 
  set of algorithmic tools available as their  undirected
  counterparts. We believe that these classes finally give a
  positive answer to the question whether interesting graph parameters
  can successfully be generalized to the directed setting. We support
  this claim by algorithmic applications described below.

  The classes of digraphs we study are classes of \emph{directed bounded expansion} and
  \emph{nowhere crownful classes} of digraphs.  These notions were defined
  in~\cite{kreutzer2012directed} where basic properties of these
  classes were developed.  The first improvement of these initial
  results appeared in \cite{kreutzer2017structural}, where classes of
  digraphs of bounded expansion were studied 
and their relation to a certain form of generalized coloring numbers
was established. 
These papers are the starting point for our
investigation in this paper. 
Furthermore, we  introduce a new type of
  digraph classes called \emph{bounded crownless expansion} which have
  the broadest set of algorithmic tools among the three types of
  digraph parameters.

  Nowhere crownful and
  directed bounded expansion classes are modeled after the
  concept of  \emph{bounded
  expansion} and \emph{nowhere denseness} 
  developed by Ne\v{s}et\v{r}il and 
  Ossona de Mendez
  ~\cite{nevsetvril2008grad,nevsetvril2011nowhere}. 
  Bounded expansion and the related concept of nowhere denseness was
  introduced to
  capture structural sparseness of graphs. On undirected
  graphs, classes of bounded expansion are very general and contain,
  for instance, planar graphs or more generally classes with excluded
  (topological) minors. But the concept goes far beyond excluded minor
  classes. 

  Following~\tcb{\cite{nevsetvril2008grad,nevsetvril2011nowhere}}, many papers have shown that
  algorithmic results for many problems on 
  classes of graphs excluding a fixed minor can be generalized to classes of
  bounded expansion \cite{bonamy2015linear, DawarK09,DrangeDFKLPPRVS16,dvovrak2013constant,
dvovrak2013testing, eickmeyer2016neighborhood,gajarsky2013kernelization, grohe2014deciding,kazana2013enumeration,kreutzer2017polynomial,
lokshtanov2015reconfiguration,nevsetvril2006linear,
segoufin2017constant}. These new algorithms not only work on much
  larger classes of graphs than those excluding a minor. Often they also become conceptually 
  simpler as they do not rely on deep, but sometimes cumbersome to use, structure theorems for classes
  with excluded minors. Furthermore, Demaine et
  al. \cite{DemaineRRVSS14} analyzed a range of real-world 
  networks and showed that many of them indeed fall within the
  framework of bounded expansion. This shows that the concept of
  bounded expansion captures many types of 
  real world instances. 
  An interesting property of classes of bounded expansion and classes
  which are nowhere dense is that they 
  can equivalently be defined in many different and seemingly unrelated
  ways: by the density of \emph{bounded depth minors}, by
  \emph{low tree depth colourings}~\cite{nevsetvril2008grad},
  by \emph{generalized coloring numbers}~\cite{zhu2009colouring}, by wideness properties such as
  \emph{uniformly quasi wideness}~\cite{nevsetvril2010first}, by \emph{sparse
    neighbourhood covers}~\cite{grohe2015colouring,grohe2014deciding},
  and many more. Each of these
  different aspects of bounded expansion classes comes with its own set
  of algorithmic tools and many of the more advanced algorithmic results on bounded
  expansion classes mentioned above crucially rely on a combination of several  of these
  techniques. 

  In this paper we study suitable generalizations of bounded expansion
  and nowhere dense  classes of graphs to the directed setting.  See
  Section~\cref{sec:definitions} for details.

  We show that  classes of digraphs of directed bounded expansion, and
  especially classes of \emph{bounded crownless expansion} that we introduce in this
paper, have very similar characterizations as their undirected counterpart: 
they have \emph{low directed tree depth colorings}, they have
\emph{bounded directed weak coloring numbers}, a concept that has been ground
breaking in the undirected setting, they have \emph{low neighborhood
complexity} and \emph{bounded VC dimension} and many more. 
As a consequence, we are able to show that most of the algorithmic
tools that were developed for undirected bounded expansion 
have their
directed counterpart. Thus, we obtain powerful algorithmic tools for
directed bounded expansion that are similar to the tools available in
the undirected setting. 

Note that we cannot combine our tools as
freely as in the undirected setting. For example, nowhere crownfulness 
does not imply that bounded depth minors are sparse, 
or directed bounded expansion does not imply directed uniform
quasi-wideness. Hence it is only natural to combine the requirements
of nowhere crownfulness with that of bounded expansion, as we
do to obtain classes of bounded crownless expansion, to obtain 
classes which behave algorithmically as nice as their undirected counterparts.

To the best of our knowledge, this is the first time that the
goal of generalizing one of the widely studied and very general
type of classes of undirected graphs to the directed setting has been
really successful and has
produced a digraph concept with a similarly broad set of
algorithmic tools as its undirected
counterpart. 
We are therefore optimistic that classes of bounded directed
expansion and classes of bounded crownfree expansion will find a wide
range of applications. We support this believe by a range of
algorithmic results we describe next.

On the more technical level, we demonstrate the power of the new
concepts by showing that several common problems on digraphs,
which do not admit efficient solutions in general, can be solved efficiently
on classes of directed bounded expansion or classes of bounded
crownless expansion. 

We first consider the \textsc{Directed Steiner Tree (\DST)} problem, 
which is defined as follows. As
  input we are given a digraph $G$, a root $r\in V(G)$, a set
  $T\subseteq V(G)\setminus \{r\}$ of terminals and an integer
  $k$. The problem is to decide if there is a set $S\subseteq
  V(G)\setminus (\{r \} \cup T)$ such that in $G[\{r\}\cup S\cup T]$
  there is a directed path from~$r$ to every terminal $T$.
 The Steiner Tree problem is one of the most intensively studied graph problems in computer
science with many important applications. We refer to the textbook of 
Pr\"omel and Steger~\cite{promel2012steiner} for more background. 
 While the parameterized complexity of Steiner Tree
parameterized by the number of terminals is well understood, not much is known about the parameterization
by the number of non-terminals in the solution tree. It is known for this parameterization that
both the directed and the undirected versions are $\mathsf{W}[2]$-hard on general graphs~\cite{molle2008enumerate}, and even on graphs of degeneracy two~\cite{Jones13}.
%,and hence unlikely to be fixed parameter tractable. 
On the positive side, it is proved 
in~\cite{Jones13} that the problem is fixed-parameter tractable when parameterized by
the number of non-terminals on graphs excluding a topological minor. This result is based 
on a preprocessing rule which allows to contract strongly connected subsets of terminal 
vertices to individual vertices. The authors furthermore show that if the subgraph induced
by the terminals is required to be acyclic, then the problem becomes fixed-parameter tractable
on graphs of bounded degeneracy. In this case, the strongly connected subsets of terminals
have diameter $0$. This suggests to consider the problem parameterized 
by the number $k$ of non-terminals plus the maximal diameter~$s$ of a strongly connected 
component in the subgraph induced by the terminals. We prove that with respect to this
parameter the problem becomes fixed-parameter tractable on every class of directed bounded
expansion. As bounded expansion classes are much more general than graphs with excluded topological minors
and are stable under bounded diameter contractions, we believe that this result may provide
the ``true'' explanation for the earlier result of~\cite{Jones13}. 

We then turn our attention to the \textsc{Distance-$r$ Dominating Set} problem. 
Given a digraph $G$ and an integer $k$, we are asked to decide whether
there exists a set $D\subseteq V(G)$ such that every vertex $v\in V(G)$
is reachable by a directed path of length at most $r$ from a vertex $d\in D$. 
The dominating set problem (and its variations) is
one of the most important problems in algorithmic graph theory. It is
NP-complete in general~\cite{karp1972reducibility}, and (under
standard complexity theoretical assumptions) cannot be approximated
better than up to a factor $\Oof(\log n)$~\cite{raz1997sub}.  This
situation is different on sparse graph classes, it admits a PTAS, e.g.,
on planar graphs~\cite{baker1994approximation}, and a constant factor
approximation on classes of undirected bounded
expansion~\cite{dvovrak2013constant}. Most generally, it admits an
$\Oof(\log k)$ approximation on graphs of bounded
VC-dimension~\cite{bronnimann1995almost} (where $k$ is the size of an
  optimal dominating set).
We study the VC-dimension of set systems
corresponding to $r$-neighborhoods in digraphs of bounded
expansion and derive an $\Oof(k\log k)$-approximation algorithm for
the \textsc{Distance-$r$ Red-Blue Dominating Set} problem and an
$\Oof(k^2\log k)$-approximation algorithm for the \textsc{Strongly Connected
Distance-$r$ Dominating Set} problem on classes of directed bounded
expansion. Our analysis is strongly based on a characterization of bounded
expansion classes in terms of generalized coloring numbers which was provided
in~\cite{kreutzer2017structural}, and which enables us to capture local separation
properties in classes of bounded expansion. 

Finally, we study classes which have both bounded expansion and which are
nowhere crownful, a property that we call \emph{bounded crownless expansion}. 
We study the kernelization problem from the \textsc{\mbox{Distance-$r$} Dominating Set}
problem. Recall that a kernelization algorithm is a polynomial-time
preprocessing algorithm that transforms a given instance into an
equivalent one whose size is bounded by a function of the parameter
only, independently of the overall input size.  We are mostly
interested in kernelization algorithms whose output guarantees are
polynomial in the parameter. The existence of a kernelization algorithm 
immediately implies that a problem is fixed-parameter tractable, and hence,
as the dominating set problem is
$\mathsf{W}[2]$-hard in general, there cannot exist a kernelization
algorithm in general (under the standard assumption that 
$\mathsf{W}[2]\neq \mathsf{FPT}$). A key ingredient to kernelization
results for dominating sets on undirected sparse graph classes~\cite{DrangeDFKLPPRVS16} 
is a duality theorem
proved by Dvo\v{r}\'ak~\cite{dvovrak2013constant}, which states
that on a graph~$G$ from a class of undirected bounded expansion the size of a minimum
distance-$r$ dominating set $\gamma_r(G)$ is only constantly larger than 
the size of a maximum packing of disjoint balls of radius $r$, $\alpha_r(G)$. 
We prove that no such duality theorem (with any functional dependence between
$\gamma_r(G)$ and $\alpha_r(G)$) can hold in graphs of directed bounded expansion. 
However, if we additionally assume that the class is nowhere crownful, we can employ
methods which were recently developed in stability theory~\cite{malliaris2014regularity}
to derive a polynomial duality theorem between domination and packing number. 
We remark that the application of stability theory in classes of bounded crownless 
expansion is not straight forward. It is known that a class of (di)graphs which is 
closed under taking subgraphs is stable, if and only if, its underlying class of undirected
graphs is nowhere dense~\cite{adler2014interpreting}. However, classes of 
bounded crownless expansion do not necessarily have this property. We have to 
carefully establish a situation in which stability is applicable, which then allows us to
derive the polynomial duality theorem. Our kernelization algorithm then follows the
approach of~\cite{DrangeDFKLPPRVS16}.

\section{Directed Minors and Directed Bounded Expansion}
\label{sec:definitions}
In this section we fix our notation. We refer
to~\cite{bang2008digraphs} for standard notation and background on 
digraph theory.

%\paragraph{Digraphs, walks and neighbourhoods.}
%A \emph{digraph} $G$ consists of a set $V(G)$ of \emph{vertices} and a
%set $E(G)\subseteq V(G)\times V(G)$ of \emph{arcs}. We assume that a
%digraph $G$ has no loops, \ie no arcs of the form $(v,v)$ for
%$v\in V(G)$.  A \emph{walk of length $k$} in a digraph $G$ is a
%sequence $W=v_0,\ldots, v_{k}$ of vertices of $G$ such that for each
%$0\leq i < k$ there is an arc $(v_i,v_{i+1})\in E(G)$. A walk is
%\emph{closed} if $v_0=v_k$, and \emph{open} otherwise. If $W$ is open,
%then vertex $v_0$ is the \emph{initial vertex} of~$W$, vertex~$v_k$ is
%its \emph{terminal vertex}, and $v_0$ and $v_k$ are
%\emph{end-vertices} of $W$. If all vertices of~$W$ are distinct,
%then~$W$ is a \emph{path from $v_0$ to~$v_k$}. If $W$ is closed and
%$v_i\neq v_0$ for all $1\leq i\leq k-1$, then it is called a
%\emph{cycle}. An acyclic digraph is called a \emph{DAG}.  A digraph
%$G$ is \emph{strongly connected} if for all distinct $u,v\in V(G)$
%there is a directed path from $u$ to $v$ and a directed path from $v$
%to $u$. The maximal strongly connected subgraphs of $G$ are its
%\emph{strongly connected components} (SCC).

Let $G$ be a digraph, let $v\in V(G)$ and let $r\geq 1$ be an integer.  The
\emph{$r$-out-neighborhood} of $v$, denoted by~$N^+_{G, r}(v)$, or
just $N^+_r(v)$ if $G$ is understood, is defined as the set of
vertices $u$ in $G$ such that~$G$ contains a directed path of length at most
$r$ from $v$ to $u$. We write $N^+(v)$ for
$N^+_1(v)\setminus\{v\}$. The \emph{$r$-in-neighborhood}~$N^-_{G,r}(v)$ 
and $N^-(v)$ are defined analogously.  The
\emph{out-degree} of a vertex $v\in V(G)$ is
$d^+(v)\coloneqq \abs{N^+(v)}$, its \emph{in-degree} is
$d^-(v)\coloneqq \abs{N^-(v)}$ and its \emph{degree} is
$d(v)\coloneqq \abs{N^+(v)}+\abs{N^-(v)}$. The \emph{minimum
  out-degree} of~$G$ is defined as
$\delta^+(G)\coloneqq\min\{d^+(v) \mid v\in V(G)\}$, \emph{minimum
  in-degree} and \emph{minimum degree} are defined analogously.
A set
$U\subseteq V(G)$ is \emph{$r$-scattered} if there is no $v\in V(G)$
and $u_1\neq u_2\in U$ with $u_1,u_2\in N_r^+(v)$.
If the arc relation of a digraph $G$ is symmetric, \ie if
$(u,v)\in E(G)$ implies $(v,u)\in E(G)$, then we speak of an
\emph{undirected graph}.  If $G$ is a digraph, we write $\bar{G}$ for
the \emph{underlying undirected graph of~$G$}, which has the same
vertices as~$G$ and for each arc $(u,v)\in E(G)$ we have
$(u,v)\in E(\bar{G})$ and $(v,u)\in E(\bar{G})$.  Note that
$\abs{E(G)}\leq \abs{E(\bar{G})} \leq 2\abs{E(G)}$.

\smallskip
\noindent \textbf{Directed minors.}
We are going to work with directed minors and directed topological
minors.  The following definition of directed minors is
from~\cite{kreutzer2012directed}.
A digraph $H$ has a \emph{directed model} in a digraph $G$ if there is
a function~$\delta$ mapping vertices $v\in V(H)$ of $H$ to sub-graphs
$\delta(v)\subseteq G$ and arcs $e\in E(H)$ to arcs
$\delta(e)\in E(G)$ such that (1) if $v\neq u$, then
$\delta(v)\cap \delta(u)=\emptyset$; (2) if $e = (u,v)$ and
$\delta(e)=(u',v')$ then $u'\in \delta(u)$ and $v'\in \delta(v)$.  For
$v\in V(H)$ let
$\mathrm{in}(\delta(v)) \coloneqq V (\delta(v)) \cap
\bigcup_{e=(u,v)\in E(H)}V(\delta(e))$
and
$\mathrm{out}(\delta(v)) \coloneqq V (\delta(v)) \cap
\bigcup_{e=(v,w)\in E(H)}V(\delta(e))$;
(3)~we require that for every $v\in V(H)$ (3.1) there is a directed path in
$\delta(v)$ from every $u\in \mathrm{in}(\delta(v))$ to every
$u'\in \mathrm{out}(\delta(v))$; (3.2) there is at least one source
vertex $s_v\in \delta(v)$ that reaches (by a directed path) every element of
$\mathrm{out}(\delta(v))$; (3.3) there is at least one sink vertex
$t_v\in \delta(v)$ that can be reached (by a directed path) from every element of
$\mathrm{in}(\delta(v))$.  We write $H\minor G$ if $H$ has a directed
model in $G$ and call $H$ a \emph{directed minor} of $G$. We call the
sets $\delta(v)$ for $v\in V(H)$ the \emph{branch-sets} of the model.

For $r\geq 0$, a digraph $H$ is a \emph{depth-$r$ minor} of a digraph
$G$, denoted as $H\minor_rG$, if there exists a directed model of $H$
in $G$ in which the length of all the paths in the branch-sets of the
model are bounded by $r$. Note that every subgraph of $G$ is a depth-$0$
minor of $G$. 

\medskip
\noindent \textbf{Directed topological minors.} 
A digraph $H$ is a \emph{topological minor} of a digraph $G$ if there
is an injective function~$\delta$ mapping vertices $v\in V(H)$ to
vertices of $V(G)$ and arcs $e\in E(H)$ to directed paths in~$G$ such
that if $e=(u,v)\in E(H)$, then $\delta(e)$ is a path from $\delta(u)$
to $\delta(v)$ in $G$ which is internally vertex disjoint from all
vertices $\delta(w)$ (for $w\in V(H)$) and all paths $\delta(e')$ (for
$e'\in E(H)$, $e'\neq e$).  For $r\geq 0$, $H$ is a \emph{topological
  depth-$r$ minor} of $G$, written $H\minor_r^tG$, if it is a
topological minor and all paths $\delta(e)$ have length at most $2r$.

\medskip
\noindent \textbf{Grads, bounded expansion and crowns.} 
Let $G$ be a digraph and let $r\geq 0$. The \emph{greatest reduced
  average density of rank~$r$} (short \emph{grad}) of $G$ is
\[\nabla_r(G)\coloneqq \max \left\{\frac{|E(H)|}{|V(H)|} \mid H\minor_r G\right\}\]
and its \emph{topological greatest average density of rank $r$} (short
\emph{top-grad}) is
\[\widetilde{\nabla}_r(G)\coloneqq \max\left\{\frac{|E(H)|}{|V(H)|}
  \mid H\minor_r^t G\right\}.\] Note that $\nabla_0(G)$ is also
  known as the \emph{degeneracy} of $G$.
  As the following theorem shows, 
  the densities of depth-$r$ minors and depth-$r$ topological minors are
functionally related.

\begin{theorem}[\cite{kreutzer2017structural}]\label{lem:densityminors}
  Let $r,d\geq 1$ and let $p=32\cdot(4d)^{(r+1)^2}$. Let $G$ be a
  digraph.  If $\nabla_r(G)\geq p$, then
  $\widetilde{\nabla}_r(G)\geq d$.
\end{theorem}

\begin{definition}
A class $\CCC$ of digraphs has \emph{bounded expansion} if there is a
function $f:\N\rightarrow \N$ such that for all $r\geq 0$ we have 
$\nabla_r(G)\leq f(r)$ (or equivalently, $\widetilde{\nabla}_r(G)\leq f(r)$) 
for all $G\in \CCC$.
\end{definition}

%Hence, classes of directed bounded expansion can equivalently be
%defined by restricting the arc density of bounded depth topological
%minors.  that is, a class $\CCC$ of digraphs has bounded expansion if
%and only if there is $f:\N\rightarrow \N$ such that for all $r\geq 0$
%it holds that $\widetilde{\nabla}_r(G)\leq f(r)$ for all $G\in \CCC$.

\vspace{-3mm}A \emph{crown} of order $q$ is a $1$-subdivision of a
clique of order $q$ with all arcs oriented away from the subdivision
vertices, that is, the digraph $S_q$ with vertex set
$\{v_1,\ldots, v_q\}\cup \{v_{ij} : 1\leq i<j\leq q\}$ and arc set
$\{(v_{ij},v_i), (v_{ij},v_j) : 1\leq i< j \leq q\}$.

\begin{definition}
A class $\CCC$ of digraphs has \emph{bounded crownless expansion} if
there is a function $f:\N\rightarrow \N$ such that for all $r\geq 0$
we have $\nabla_r(G)\leq f(r)$ and $S_{f(r)}\not\minor_r G$ for
all $G\in \CCC$.
\end{definition}

%%% Local Variables:
%%% mode: latex
%%% TeX-master: "soda"
%%% End:

\section{Steiner trees}

\begin{definition}\label{def:dst}
  The \textsc{Directed Steiner Tree} (\DST) problem is defined as
  follows. The input is a tuple $(G,r,T,k)$ where $G$ is a digraph, $r\in V(G)$ is
    a vertex (a root), a set $T\subseteq V(G)\setminus \{r\}$ of terminals and $k$ is an integer.
The problem is to decide if there is a set $S\subseteq
  V(G)\setminus (\{r \} \cup T)$ of size at most $k$ such that in $G[\{r\}\cup S\cup T]$
  there is a directed path from $r$ to every terminal $T$.
\end{definition}

The $\DST$ problem has been widely studied in the area of approximation
algorithms as it generalizes several routing and
domination problems.  We are interested in the parameterized
complexity of this problem. It follows from an algorithm by Nederlof
\cite{Nederlof09} and Misra et al. \cite{Misra10}, that the problem
can be solved in time 
$2^{|T|}\cdot p(n)$, for some polynomial $p(n)$. 
In this paper, we are interested in the
standard parameterization in parameterized complexity, where as
parameter we take the solution size, i.e. we take 
the number $k$ of non-terminals as parameter. This models the case
where we need to pay for any node we add to the solution and we want
to keep the bound $k$ on these nodes as small as possible without any
restriction on the number of terminals to connect.

In \cite{Jones13}, Jones et al. show that $\DST$ with this
parameterization is fixed-parameter tractable on any class of digraphs
such that the class of underlying undirected graphs excludes a fixed
graph $H$ as an undirected topological minor. In this section we show that this
result can be extended to classes of bounded directed 
expansion, but with one restriction on the structure of the terminals.

\begin{theorem}\label{thm:steiner}
  Let $\CCC$ be a class of digraphs of bounded expansion. $\DST$ is fixed
  parameter tractable on $\CCC$ parameterized by the number $k$ of
  non-terminals in the solution plus the maximal diameter $s$ of the
  strongly connected components in the subgraph induced by the
  terminals.
\end{theorem}
\begin{proof}
  Let $G'\in\CCC$ be a digraph, $r\in V(G')$ be the root node,
  $T'\subseteq V(G')\setminus \{r\}$ be the set of terminals and let $k\geq 0$
  be an integer. Let $s$ be the maximal diameter of a strongly
  connected component of $G'[T']$.

  As a first reduction step, we contract every strongly connected
  component of $G'[T']$ into a single vertex. Let $G$ be the resulting
  digraph and let $T$ be the set of terminals which, for every strong
  component of $G'[T']$ contains the corresponding new vertex obtained
  by contraction.  It is easily seen that any set $S\subseteq V(G)$ is
  a solution of $(G, r, T, k)$ if, and only if, $S$ is a solution of
  $(G', r, T', k)$.  Let $d \coloneqq 2\nabla_0(G)$. 
  Note that $G$ is a depth-$s$ minor of $G'$, as every strong
  component that was contracted to obtain $G$ has diameter at most
  $s$. It follows that $d \leq 2\nabla_s(G')$ is bounded by a
  constant only depending on $s$ (and the expansion of $\CCC$).

  Let $T_0\subseteq T$ be the set of terminals that have in-degree $0$
  in $G[T]$. Since for every
terminal $t\in T$, the graph $G[T]$ contains a path from some 
$t_0\in T_0$ to $t$, we have for every set $S\subseteq V(G)$ 
the property that in $G[\{r\}\cup S\cup T]$ 
there is a directed path from $r$ to every 
$t\in T$ if and only if there is a 
directed path from $r$ to every $t_0\in T_0$. We now prove the following claim (see \cite[Lemma 2]{Jones13}).

  \begin{Claim}\label{claim:nederlof}
    There is an algorithm which, given a digraph $G$,
    $r\in V(G), T\subseteq V(G)\setminus \{r\}$ and $T_0\subseteq T$, computes a minimum size set $S\subseteq V(G)$
    such that there is a path from $r$ to every $t\in T_0$
    in $G[\{r\}\cup T\cup S]$ in time $2^{|T_0|}\cdot p(n)$, for some
    fixed polynomial $p(n)$ where $n=\abs{V(G)}$. 
\end{Claim}

\begin{ClaimProof}
  We first add an edge from every $t\in T_0$ to every node
  $u\in V(G)\setminus T$ which is reachable from $t$ by a path whose
  internal vertices are all in $T$. Let $G''$ be the resulting
  graph. Then any set $S$ is a solution for 
  the problem from the claim if, and
  only if, it is a solution for $\textsc{Dst}(G'', r, T_0, k)$. We can now call the algorithm of Misra et al.
  \cite{Misra10}, to solve the instance $\textsc{Dst}(G'', r, T_0, k)$ in time
  $2^{|T_0|}\cdot p(n)$, where $p(n)$ is
  a fixed polynomial. This immediately implies the claim.
\end{ClaimProof}

We are now ready to present a recursive algorithm for deciding whether
$\textsc{Dst}(G, r, T, k)$ has a solution. In the recursive calls we
are additionally given a partial solution {$Y\subseteq V(G)\setminus (T\cup\{r\})$} as input and we need to
decide if $\textsc{Dst}(G, r, T, k)$ has a solution extending
$Y$. {Note that a solution extending $Y$ is simply a solution of $\textsc{Dst}(G, r, T\cup Y, k-|Y|)$.} 
%Then, to solve $\textsc{Dst}(G, r, T, k)$ we simply call the
%algorithm with $Y=\emptyset$. 

Let $G, r, T, k, Y$ be given as above. Clearly, if $|Y| > k$ we
can reject immediately. Recall that $d = 2\nabla_0(G)$.  Let
$N \coloneqq V(G) \setminus (T\cup Y\cup \{r\})$ be the set of non-terminal
vertices.
Let $k_Y \coloneqq k - |Y|$ be the number of non-terminals we can still
choose for our solution.  Let {$T_Y \coloneqq \{ t\in T \sth t \in N^+(Y\cup \{r\}) \}$}
be the set of terminals dominated by $Y$ {or $r$}. Let
$\bty \coloneqq T_0\setminus T_Y$. 
%As {argued} above, it remains {to solve the instance $\textsc{Dst}(G, r, T\cup Y, k-|Y|)$.}

Let
$\Sh \coloneqq \{ v \in N \sth |N^+(v)\cap \bty| > d \}$ be the set of
non-terminals dominating more than $d$ elements in $\bty$ and let
$\Sl \coloneqq N\setminus \Sh$ be the non-terminals dominating at most $d$
elements of $\bty$.
Similarly, let $\Th \coloneqq \{ t\in \bty \sth t\in N^+(\Sh)\}$ be the set
of source terminals dominated by $\Sh$ and let
$\Tl \coloneqq \{ t\in \bty \sth t \not\in N^+(Y\cup \Th)\}$ be the set of
source terminals not dominated by either $Y$ or $\Th$.

Clearly, if $|\Tl|> d\cdot k_Y$, then we can reject the input, as all
vertices in $\Tl$ can be dominated only by non-terminals from $\Sl$
and each $v\in \Sl$ can dominate only $d$ elements of $\Tl$.  Hence we
can assume that $|\Tl|\leq d\cdot k_Y$. It follows that if $\Sh$ is
empty and therefore $\Th$ is empty, we can apply the algorithm in
Claim~\ref{claim:nederlof} to decide whether {$\DST(G,r,T\cup Y,k_Y)=\DST(G,r,\Tl\cup Y,k_Y)$} has a solution
extending $Y$ in time $2^{\abs{T_{\le d}}}\cdot p(n) \le 2^{d\cdot
    k_Y}\cdot p(n)$.

Thus, we can assume that $|\Tl|\leq d\cdot k_Y$ and
$\Sh\not=\emptyset$. Now choose among all vertices in $\Th$ a vertex
$v\in \Th$ which minimizes $d_v \coloneqq |\Sh\cap N^-(v)|$. We claim that
$d_v \leq d$. To see this, consider the subgraph~$H$ of $G$ with
vertex set $\Sh\cup \Th$ and all edges of $G$ from $\Sh$ to $\Th$. As
this is a subgraph of $G$ it follows that
$2\nabla_0(H)\leq 2\nabla_0(G) = d$. Hence, there must be a vertex in
$\Th$ of in-degree at most $d$ and therefore, by the choice of $v$, we
have $d_v\leq d$.

Clearly, any solution to $\DSTG$ extending $Y$ must contain a vertex
dominating $v$. 
%This can either be one of the high degree vertices in
%$\Sh$ or, if none of them was chosen for the solution,  then some other
%node dominating $v$ must be contained in the solution. 
We can
explore all possibilities by branching into $d+1$ recursive calls: for
each $s\in (\Sh\cap N^-(v))$ we call {$\DST(G,r,T\cup Y\cup\{s\}, k_Y-1)$} recursively. If one of these calls is successful, then we return the
solution. Otherwise we know that there is no solution in which $v$ is
dominated by a high degree vertex and hence, {$\DST(G,r,T\cup Y,k_Y)$} has a solution
if, and only if, {$\DST(G-(\Sh\cap N^-(v)), r, T\cup Y, k_Y)$} has
a solution.  Hence, the last branch is to recursively
call {$\DST(G-(\Sh\cap N^-(v)), r, T\cup Y, k_Y)$.} If this returns a
solution extending $Y$, we are done and return this solution,
otherwise we can reject the input.  Note that in this recursive
instance, $v$ is no longer a vertex in $\Th$ as it is not dominated
anymore by any non-terminal which dominates at least $d$ 
nodes in $\bty$.

We show next that the algorithm terminates {sufficiently fast} on all inputs. In every
recursive call, either the set $Y$ increases and $|\Tl|$ does not
change or $Y$ is not changed but $|\Tl|$ increases by one, as the
vertex~$v$ in the last branch is now added to $\Tl$ in the recursive
call. For every node $x$ in the recursion tree we can therefore define
the complexity of $x$ as $|Y|+|\Tl|$, where $Y$ and $\Tl$ are the
corresponding sets of the $\DST$-instance solved at this node. Hence,
every recursive call increases the complexity.

As the algorithm terminates as soon as $|\Tl|> d\cdot (k-|Y|)$ or
$|Y|>k$, this means that every branch of the recursion tree has length
at most $k+d\cdot k = k\cdot (d+1)$. As at every node we do at most
$d+1$ recursive calls, this means that the entire search tree has at
most $(d+1)^{k\cdot (d+1)} = 2^{k\cdot (d+1) \cdot \log (d+1)}$ nodes.
Clearly, the computation at every node can be done in polynomial
time. Hence, the entire algorithm runs in time
$\Oof(2^{k\cdot (d+1)\cdot \log (d+1)})\cdot p(n)$ for some fixed
polynomial $p(n)$.  This completes the proof of the theorem.
\end{proof}

The proof of the theorem has the following immediate consequences.

\begin{corollary}
  Let $\CCC$ be a class of digraphs closed under taking directed
  minors for which $\nabla_0(G) \leq c$ for a constant $c$ for all
  $G\in \CCC$. Then $\DSTG$ can be solved for all $G\in \CCC$, $r\in
  V(G)$, $T\subseteq V(G)\setminus \{r\}$ and~$k$ in time
  $2^{\Oof(k)}\cdot p(n)$, for some fixed polynomial $p(n)$.
\end{corollary}

Note that this strictly generalizes classes of undirected graphs
excluding a fixed minor. 

Another consequence of this is the following result, which immediately
follows from the well-known observation in parameterized complexity
(see e.g. \cite[Lemma 7]{Jones13}), that for all functions $g(n) = o(\log n)$
there is a function $f(k)$ such that $f(k) \leq 2^{g(n)\cdot k}$, for all $k$ and all $n$.

\begin{corollary}
  Let $\CCC$ be a class of digraphs such that $\nabla_{|G|}(G)\cdot \log \nabla_{|G|}(G) \leq o(\log
  n)$  for all $G\in \CCC$. Then $\DST$ is fixed-parameter tractable on $\CCC$ with parameter $k$.
\end{corollary}

Finally, the result also implies an fpt factor-$2$-approximation
algorithm for the Strongly Connected Steiner Subgraph problem,
\textsc{Scss}, on classes of bounded directed expansion. In the
\textsc{Scss} we are given a digraph $G$, a number $k$, and a set $T$
of terminals and we are asked to compute a set $S$ of at most $k$
non-terminals such that $G[T\cup S]$ is strongly connected. 

\begin{theorem}
  Let $\CCC$ be a class of digraphs of bounded expansion. 
  There is an fpt factor-$2$-approximation algorithm for
  $\textsc{Scss}$ on $\CCC$ parameterized by the number $k$ of 
  non-terminals in the solution plus the maximal diameter $s$ of a 
  strongly connected component in the subgraph of $G$ induced by the
  terminal nodes.
\end{theorem}
\begin{proof}
  Note first that if $H$ is obtained from a
  digraph $G$ by reversing the 
  orientation of all edges of $G$, then for all $r$, $\nabla_r(H) =
  \nabla_r(G)$. Now, given a digraph $G$, a number $k$ and a set $T$ of
  terminals, we can fix a terminal $t\in T$ and solve $P_1 \coloneqq \DST(G, t,
  T\setminus\{t\}, k)$ and $P_2 \coloneqq \DST(H, t, T\setminus\{t\}, k)$, where
  $H$ is obtained from $G$ by reversing all edges. We then take the union of
  the two solutions $S_1$ and $S_2$ for $P_1$ and~$P_2$. Clearly,
  if $\textsc{Scss}(G, k, T)$ has a solution $S$ of size $k$ then $S$ is
  also a solution for the two subproblems. Hence, $|P_1|, |P_2|\leq k$
and therefore $|S_1\cup S_2|\leq 2k$ as claimed. 
\end{proof}
We close the section by showing that for bounded expansion classes,
the parameterization $k+s$ in \cref{thm:steiner} cannot be replaced by
taking only~$k$ as parameter. This follows immediately from a result
of~\cite{Jones13} where it is shown that \textsc{Set Cover} can be
reduced to $\DST$ on $2$-degenerate graphs. It is straight forward to
modify this example so that the resulting class of graphs has bounded
directed expansion.

\begin{theorem}
  The $\DST$-problem restricted to classes of digraphs of bounded
  expansion parameterized by the solution size $k$ is $W[2]$-hard.
\end{theorem}

%%% Local Variables:
%%% mode: latex
%%% TeX-master: "soda"
%%% End:

\section{VC-dimension and domination}\label{sec:vc}

We come to another algorithmic application on graphs of directed 
bounded expansion, namely, the approximation of the \textsc{Distance-$r$
Dominating Set} problem. We study the VC-dimension of set systems
corresponding to $r$-neighborhoods in digraphs of bounded
expansion and derive an $\Oof(k\log k)$-approximation algorithm for
the \textsc{Distance-$r$ Red-Blue Dominating Set} problem and an
$\Oof(k^2\log k)$-approximation algorithm for the \textsc{Strongly Connected
Distance-$r$ Dominating Set} problem on classes of directed bounded
expansion.

\subsection{VC-dimension and neighborhood complexity}

Let $\FFF\subseteq 2^A$ be a family of subsets of a set $A$. For a set
$X\subseteq A$, we denote $X\cap \FFF=\{X\cap F : F\in \FFF\}$.  The
set $X$ is \emph{shattered by $\FFF$} if $X\cap \FFF=2^X$.  The
\emph{Vapnik-Chervonenkis dimension}, short \emph{VC-dimension}, of~$\FFF$ is the maximum size of a set $X$ that is shattered by $\FFF$.

Note that if $\FFF$ has VC-dimension $d$, then also
$B\cap \FFF$ for every subset $B\subseteq A$ of the ground set has VC-dimension
at most $d$.
The following theorem was first proved by Vapnik and
Chervonenkis~\cite{vapnik1971uniform}, and rediscovered by
Sauer~\cite{sauer1972density} and
Shelah~\cite{shelah1972combinatorial}. It is often called the
Sauer-Shelah lemma in the literature.
\begin{theorem}\label{thm:sauer_shelah}
  If $|A|\leq n$ and $\FFF\subseteq 2^A$ has VC-dimension $d$,
  then
  $|\FFF|\leq \sum_{i=0}^{d}\binom{n}{i}\in \Oof(n^d)$.
\end{theorem}

\begin{definition}
In the
\textsc{Distance-$r$ Red-Blue Dominating Set} problem, we are given a
digraph $G$ and two sets $R,B\subseteq V(G)$ and an integer $k$, and
asked whether there exists a subset $D\subseteq B$ of at most $k$
\emph{blue} vertices such that each \emph{red} vertex from $R$ is at
distance at most~$r$ to a vertex in $D$.  We allow that $R$ and $B$
intersect.
\end{definition}

The study of the distance-$r$
dominating set problem in context of bounded VC-dimension motivates
the following definition. 
Let $G$ be a digraph and $r\geq 1$. The \emph{distance-r VC-dimension}
of $G$ is the VC-dimension
of the set family $\{N_r^-(v)\colon v\in V(G)\}$ over the set $V(G)$.

According to \cref{thm:sauer_shelah}, the distance-$r$ VC-dimension of
a graph  {is} bounded if the \emph{distance-$r$ neighborhood
  complexity} of its sets is polynomially bounded.  Let $G$ be a
digraph, let $X\subseteq V(G)$ and let $r\geq 1$. The
\emph{distance-$r$ neighborhood complexity of $X$ in $G$}, denoted
$\nu^-(G)$, is defined by
\[\nu^-(G,X) \coloneqq \abs{\{N^-_r(v)\cap X \mid v\in V(G)\}}.\]

Analogously, one can define the \emph{distance-$r$ out-neighborhood
  complexity} when using $N_r^+(v)$ and the \emph{distance-$r$ mixed
  neighborhood complexity} when using $(N_r^+(v)\cup N_r^-(v))$ in
the above definition and our proofs can be analogously carried out for
these measures.

It was proved in~\cite{reidl2016characterising} that a class $\CCC$ of
undirected graphs has bounded expansion, if and only if, for every
$r\geq 1$ there is a constant $c_r$ such that for all $G\in \CCC$ and
all $X\subseteq V(G)$ we have $\nu(G,X)\leq c_r\cdot |X|$.  The
analogous statement for classes of directed graphs does not hold, not
even for $r=1$, as pointed out in~\cite{kreutzer2017structural}.
However, we prove that the distance-$r$ neighborhood complexity of a
digraph can be bounded in terms of its weak $r$-coloring numbers. 

The weak $r$-coloring numbers for undirected graphs were introduced
by Kierstead and Yang~\cite{kierstead03} and they play a key role 
in the algorithmic theory of graphs of undirected bounded expansion, 
ever since Zhu~\cite{zhu2009colouring} proved that these classes
can be characterized by the weak coloring numbers. 

Let $G$ be a digraph.  By $\Pi(G)$ we denote the set of all linear
orders of $V(G)$.  For $r\geq 0$, we say that $u$ is \emph{weakly
  $r$-reachable} from~$v$ with respect to an order~$L\in \Pi(G)$ if
there is a path $P$ of length at most~$r$, connecting $u$ and $v$, \emph{in
either direction}, such that $u$ is minimum among the vertices of $P$
with respect to $L$. By $\dWReach{r}[G,L,v]$ we denote the set of
vertices that are weakly $r$-reachable from~$v$ with respect to~$L$.
We define the \emph{weak $r$-coloring number $\dwcol{r}(G)$} of~$G$
as
\begin{align*}
  \dwcol{r}(G) & \coloneqq  \min_{L\in\Pi(G)}\:\max_{v\in V(G)}\:
                 \bigl|\dWReach{r}[G,L,v]\bigr|\,.
\end{align*}
Note that $\dwcol{r}(G)$ is a {\em
  monotone parameter}, in the sense that if $H\subseteq G$, then $\dwcol{r}(H)\leq \dwcol{r}(G)$. 

\begin{theorem}[\cite{kreutzer2017structural}]\label{thm:wcol-be}
  A class $\CCC$ of digraphs has bounded expansion if, and only if,
  there is $f\colon\N\rightarrow\N$ such that $\dwcol{r}(G) \leq f(r)$
  for all $G\in \CCC$ and all $r\geq 1$.
\end{theorem}

The next lemma shows that the weak $r$-coloring numbers are very
useful to describe local separation properties in graphs of bounded
expansion. 

\begin{lemma}\label{lem:wreach-sep}
  Let $G$ be a digraph and let $r\geq 1$. Let $P$ be a path of length
  at most $r$ with endpoints $u$ and $v$ in either direction. Let $L$ be an order of
  $V(G)$. Then $\dWReach{r}[G,L,u]\cap \dWReach{r}[G,L,v]$ contains a
  vertex of~$P$.
\end{lemma}
\begin{proof}
  Let $z$ be the minimal vertex of $P$ with respect to $L$. Then
  $z\in \dWReach{r}[G,L,u]$ and $z\in \dWReach{r}[G,L,v]$.
\end{proof}

Using \cref{lem:wreach-sep} we can well control the interaction 
of distance-$r$ neighborhoods with a set $X$.  Let~$G$ be a
digraph and let~$L$ be a linear order on $V(G)$ and let $r\geq 1$.
Let $A\subseteq V(G)$ be enumerated as $a_1,\ldots, a_{|A|}$,
consistently with the order.  For $v\in V(G)$ let $D_r^-(v,A)$ denote
the \emph{distance-$r$ vector} of $v$ and~$A$, that is, the vector
$(d_1,\ldots, d_{|A|})$, where $d_i=\dist(a_i,v)$ if
$0\leq \dist(a_i,v)\leq r$, and $\infty$ otherwise. Here $\dist(a_i,v)$ is
the length of a shortest path from $a_i$ to $v$.

\begin{lemma}\label{lem:nc}
  Let $G$ be a digraph, let $X\subseteq V(G)$ and let $r\geq 1$.  Let
  $c\coloneqq \dwcol{r}(G)$.  Then the number of distinct distance-$r$
  vectors $D_r^-(v,X)$ is bounded by {$((r+2)\cdot c\cdot |X|)^c$}, and in
  particular,
\[\nu_r^-(G,X)\leq ((r+2)\cdot c\cdot |X|)^c.\]
\end{lemma}
\begin{proof}
  Let
  $W\coloneqq \dWReach{r}[G,L,X]=\bigcup_{x\in X}\dWReach{r}[G,L,x]$.
  We claim that if
  \[D_r^-(u, W\cap \dWReach{r}[G,L,u])=D_r^-(v, W\cap
  \dWReach{r}[G,L,v]),\]
  then \[N_r^-(u)\cap X = N_r^-(v)\cap X.\]  To see this, fix $u$ and
  $v$ with
  $D_r^-(u, W\cap \dWReach{r}[G,L,u])=D_r^-(v, W\cap
  \dWReach{r}[G,L,v])$
  and some $x\in X$. Assume that $x\in N_r^-(u)\cap X$. We prove that
  $x\in N_r^-(v)\cap X$. Let $P$ be a shortest path from $x$ to $u$.
  By \cref{lem:wreach-sep}, $P$ contains a vertex $z$ of
  $W\cap \dWReach{r}[G,L,u]$ and because $P$ is a shortest path, the
  subpath of $P$ between $z$ and $u$ is of minimal distance, say
  distance $r'\leq r$.  Because
  $D_r^-(u, W\cap \dWReach{r}[G,L,u])=D_r^-(v, W\cap
  \dWReach{r}[G,L,v])$,
  also $v$ is at distance at most $r'$ to $z$.  Then the initial part
  of $P$ from $x$ to $z$ together with the path from $z$ to $v$
  witnesses that $x\in N_r^-(v)\cap X$.  The case that
  $x\in N_r^-(v)\cap X$ is symmetric.

  Now, since $\abs{W}\leq c\cdot |X|$ and we have
  $\abs{W\cap \dWReach{r}[G,L,v]}\leq c$ for all $v\in V(G)$, we have
  {$\abs{\{D_r^-(v, W\cap \dWReach{r}[G,L,v]) : v\in V(G)\}}\leq
  (c\cdot \abs{X})^c \cdot (r+2)^c$.}
  As argued above, this number of distinct distance profiles bounds
  the number of neighborhoods in $\nu^-_r(G,X)$.
\end{proof}

%All proofs marked with $\star$ can be found in the appendix. 

\begin{corollary}\label{crl:be-vc}
  Let $G$ be a digraph and $r\geq 1$. Then the distance-$r$
  VC-dimension of $G$ is bounded by {\mbox{$(r+2)\cdot 
  (2 \dwcol{r}(G))^2$}.}
\end{corollary}
\begin{proof}
{Let $c\coloneqq \dwcol{r}(G)$ and let $X\subseteq V(G)$ be the largest set which is shattered by
$\{N_r^-(v) : v\in V(G)\}$. 
Let $\FFF=\{N_r^-(v)\cap X :v\in V(G)\}$. As
$X$ is shattered, $\FFF$ has VC-dimension $|X|$
and contains $2^{|X|}$ elements. 
On the other hand, according to \Cref{lem:nc} we have 
$\nu_r^-(G,X)\leq ((r+2)\cdot c\cdot |X|)^c$.
Hence $\FFF$ has at most $((r+2)\cdot c\cdot |X|)^c$ elements, 
which implies that $2^{|X|}\leq ((r+2)\cdot c\cdot |X|)^c$. 
Assuming $|X|\geq ((r+2)\cdot c)$, we get
\[2^{|X|}\leq |X|^{2c} \Leftrightarrow |X|/\log|X|\leq 2c\]
which is violated for $|X|\geq (2c)^2$, as $c\geq 2$ (unless $G$ 
is an edgeless graph in which case the claim trivially holds).  }
\end{proof}

\subsection{Approximation of distance-r red-blue dominating sets}

For our approximation algorithm we will make use of the following
algorithm of Br\"onnimann and Goodrich.

\begin{theorem}[Br\"onnimann and Goodrich~\cite{bronnimann1995almost}]\label{thm:bronnimann}
  For every fixed dimension $d\geq 1$, there is a polynomial-time
  algorithm for finding a hitting set in a set system $\FFF$ of
  VC-dimension $d$ of size $\Oof(k\cdot \log k)$, where $k$ is the
  size of a minimum hitting set for $\FFF$.
\end{theorem}

\begin{theorem}\label{crl:red-blue-ds}
  Let $\CCC$ be a class of bounded expansion and let $r\geq 1$. There
  is a polynomial time algorithm which on input
  $G\in \CCC, R,B\subseteq V(G)$ computes a distance-$r$ red-blue
  dominating set of $G$ of size $\Oof(k\cdot \log k)$, where $k$ is
  the size of a minimum distance-$r$ red-blue dominating set in $G$.
\end{theorem}
\begin{proof}
Let $B\subseteq V(G)$ and $R\subseteq V(G)$. Let
$\mathcal{F}=\{N_r^-(v)\cap B : v\in R\}$. Then a hitting set
of~$\mathcal{F}$ is a blue distance-$r$ dominating set of
$R$. According to \cref{crl:be-vc}, $\mathcal{F}$ has bounded
VC-dimension on any bounded expansion class. Now conclude with
\cref{thm:bronnimann}.
\end{proof}

With slightly more effort we can approximate the connected distance-$r$ dominating
set problem. 

\begin{theorem}\label{thm:approx-cds}
  Let $\CCC$ be a class of bounded expansion and let $r\geq 1$. There
  is a polynomial time algorithm which on input $G\in \CCC$ computes a
  strongly connected distance-$r$ dominating set of $G$ of size
  $\Oof(k^2\cdot \log k)$, where $k$ is the size of a minimum strongly
  connected distance-$r$ dominating set of $G$.
\end{theorem}
\begin{proof}
  Assume we know the size $k$ of an optimal strongly connected
  distance-$r$ dominating set (we will incrementally {test} all values
  $1,\ldots, k$ until we find a valid solution).  We guess one vertex
  $v\in V(G)$ which is a central vertex of an optimal strongly
  connected distance-$r$ dominating set $D$. Observe that $D$ has
  radius at most $k$, and hence, we can restrict our search for an
  approximate solution to the strongly connected $k$-neighborhood of
  $v$, which we color blue. We color the rest of the graph red and
  now search for an approximate distance-$r$ red-blue dominating set
  using \cref{crl:red-blue-ds}. We find a solution $D'$ of size
  $\Oof(k\cdot \log k)$, which may not be connected. {Because we
  restricted the blue vertices to the strong $k$-neighborhood of $v$,
  for each $w\in D'$ there is a closed walk $W_{v,w}$ 
  of length at most $k$ which contains both~$v$ and $w$. Now
  taking the union of all vertices of the walks $W_{v,w}$ for
  $w\in D'$ gives us a dominating set of size at most $k$ times
  larger than $D'$.} Hence,
  we compute an $\Oof(k^2\cdot \log k)$ approximation to a strongly
  connected distance-$r$ dominating set.
\end{proof}

\section{Kernelization on classes of bounded crownless expansion}

A powerful method in parameterized complexity is
\emph{kernelization}. A kernelization algorithm is a polynomial-time
preprocessing algorithm that transforms a given instance into an
equivalent one whose size is bounded by a function of the parameter
only, independently of the overall input size.  We are mostly
interested in kernelization algorithms whose output guarantees are
polynomial in the parameter. As the dominating set problem is
$\mathsf{W}[2]$-hard in general, we cannot expect a kernelization
algorithm in general. Again, the situation is quite different in
sparse graphs. The \textsc{Dominating Set} problem admits linear kernels
on planar graphs~\cite{alber2004polynomial}, bounded genus
graphs~\cite{bodlaender2009meta}, apex-minor free graphs~\cite{fomin2010bidimensionality},
$H$-minor free graphs~\cite{fomin2012linear} and $H$-topological minor free
graphs~\cite{fomin2013linear}. It admits polynomial kernel{s}
on graphs of bounded degeneracy~\cite{philip2012polynomial}.
The \textsc{Distance-$r$ Dominating Set} problem admits a linear kernel
on classes of bounded expansion~\cite{DrangeDFKLPPRVS16}, 
and almost linear kernel on nowhere dense classes of 
graphs~\cite{kreutzer2017polynomial}. We are not aware of 
any kernelization results on directed graphs, though, it is easy
to observe that the result of~\cite{philip2012polynomial}
also holds on digraphs of bounded degeneracy. 

We prove that for every fixed value of $r\geq 1$, the
distance-$r$ dominating set problem admits a polynomial kernel on
every class of bounded crownless expansion. For this, we first prove
in \cref{sec:duality} a polynomial duality theorem between the size of
a largest $r$-scattered set and a smallest distance-$r$ dominating set
in these classes.  We show that this duality theorem does not hold in
classes of directed bounded expansion without the additional
assumption on crown-freeness. In \cref{sec:kernel} we then adapt a
method developed in~\cite{DrangeDFKLPPRVS16} for kernelization on
classes of undirected bounded expansion to the directed case.

\subsection{A polynomial duality theorem}\label{sec:duality}

Denote by $\gamma_r(G)$ the size of a smallest set $X$ such that
$N_r^+(X)=V(G)$. Denote by $\alpha_r(G)$ the size of a largest set $Y$
such that for all $x,y\in Y$ there is no $u\in V(G)$ with
$x,y\in N_r^+(u)$, that is, the largest set which is
$r$-scattered. Clearly, $\gamma_r(G)\geq \alpha_r(G)$, because no
vertex in $G$ can $r$-dominate more than one vertex of $Y$. As shown
by Dvo\v{r}\'ak~\cite{dvovrak2013constant}, in a class $\CCC$ of
undirected bounded expansion, it holds that
$\gamma_r\leq c\cdot \alpha_r(G)$ for some constant $c$ depending only
on $\CCC$. This is not true for directed graphs.

\begin{theorem}
There is a class of directed bounded expansion such that for every constant $c$ we have 
$\gamma_1(G)\geq c$ for infinitely many $G\in \CCC$ and $\alpha_1(G)={2}$ for all $G\in \CCC$. 
\end{theorem}
\begin{proof}
  Let $n\in \N$. Let $G_n$ be the digraph with vertex set
  $\{v_1,\ldots, v_n\}\cup \{w_{ij} : 1\leq i<j\leq n\}\cup\{a\}$ and
  arc set
  $\{(w_{ij},v_i), (w_{ij},v_j) : 1\leq i<j\leq n\}\cup \{(a,w_{ij}) :
  1\leq i<j\leq n\}$,
  that is, $G_n$ is obtained from a $1$-subdivision of a clique of
  size $n$ with all subdivision arcs pointing away from the
  subdivision vertex, together with an apex vertex adjacent to all
  subdivision vertices. Then
  $\gamma_1(G_n)=\left\lceil n/2\right\rceil +1$. Every subdivision
  vertex can dominate $2$ principle vertices. The apex vertex
  dominates all subdivision vertices, and this is best possible. We
  have $\alpha_1(G)={2}$, as for all $x,y\in V(G_n){\setminus\{a\}}$ we either have
  $(x,y)\in E(G_n)$, or there is $u\in V(G_n)$ with $x,y\in N(u)$.
\end{proof}

If we however have a class of graphs which has bounded crownless
expansion, then $\gamma_r$ and $\alpha_r$ are polynomially related. We
will apply the algorithm of Dvo\v{r}\'ak~\cite{dvovrak2013constant} to
the digraph $G$ and prove that it finds both a dominating set and a
polynomially smaller independent set. Our
proof is inspired by a recent result of Malliaris and Shelah on stability 
theory~\cite{malliaris2014regularity}, which allows us to apply 
a Ramsey type argument{, but} with polynomial instead of exponential 
bounds. 

Let $T$ be a (rooted) binary tree, where each vertex (except the root)
is marked as a left or right successor of its predecessor. We call $w$
a \emph{left (right) descendant} of $v$ if the first successor on the
unique $v$-$w$ path in $T$ is a left (right) successor.

Fix an enumeration $a_1,\ldots, a_{\ell}$ of a set $A{\subseteq V(G)}$.  The
\emph{$r$-independence tree} of $(a_1,\ldots,a_{\ell})$ is a binary
tree which is constructed recursively as follows. We make $a_1$ the
root of the tree. Assume that $a_1,\ldots, a_i$ have already been
inserted into the tree. In order to insert the next element $a_{i+1}$,
we follow a root-leaf path to find a position for it.  Starting from
the root $a_1$, at each point we are at some node $a_j$ and we {have} to
decide whether we continue along the left or to the right branch at
$a_j$.  If there is an element $u$ such that
$a_j,a_{i+1}\in N_r^+(u)$, we continue along the right branch at
$a_j$, otherwise we follow the left branch. If there is no right
successor (or left successor, respectively), we insert $a_{i+1}$ as a
right (or left child, respectively) of $a_j$.

\begin{lemma}\label{lem:number-of-nodes}
  Let $T$ be a rooted binary tree and let $t\geq 1$ be an integer. 
  Assume that no root-leaf path in~$T$ contains
  a sub-sequence $a_1,\ldots, a_t$ (of pairwise distinct elements)
  such that $a_j$ is a right
  descendant of $a_i$ for all $1\leq i<j\leq t$. If $T$ has height at
  most~$h$, then $T$ has at most $h^{t+1}$ vertices.
\end{lemma}
\begin{proof}
  We can describe the position of each node $u$ of $T$ by a set
  $P\subseteq \{0,\ldots, h\}$ such that $h(u)\in P$, where $h(u)$
  denotes the height of $u$ in~$T$, and $h(w)\in P$ for all $w$ such
  that $u$ is a right descendant of $w$ in~$T$.  The position of $u$
  in $T$ is then found by following a path from the root which turns
  right at the smallest $|P|-1$ levels which are contained in $P$ and
  which stops at the largest level in $P$ (which is the number
  $h(u)$). It hence suffices to count the number of possible such sets
  $P$. Since by assumption, no path in $T$ contains a sub-sequence
  $a_1,\ldots, a_t$ such that $a_j$ is a right descendant of $i$ for
  all $1\leq i<j\leq t$, every set $P$ describing a position in $T$
  has at least $1$ and at most $t+1$ entries.  We conclude that there
  are at most
\[\sum_{i=1}^{t+1}\binom{h}{i}\leq h^{t+1}\]
elements in $T$. 
\end{proof}

Unfortunately, we cannot completely avoid the usage of Ramsey 
arguments, however, the numbers will be fixed constants depending only on the radius $r$, the density $c$ at depth $r$ and the order of the crown that we assume is excluded at depth $r$. 
\pagebreak

\begin{lemma}[Finite Canonical Ramsey Theorem \cite{erdos1950combinatorial}]
\label{thm:CR}
{For every integer $k$ there exists an integer $n$ with the following property:
	Given any $f:[n]\times [n]\rightarrow \mathbb N$, there exists an subset
        $C\subseteq [n]$ of size $k$ such that either{ of the following holds.}}
	\begin{enumerate}
		\item {For all $a_1,b_1,a_2,b_2\in C$, we have $f(a_1,b_1)=f(a_2,b_2)$.}
		\item {For all $a_1,b_1,a_2,b_2\in C$, we have $f(a_1,b_1)=f(a_2,b_2)\Leftrightarrow a_1=a_2$.}
		\item {For all $a_1,b_1,a_2,b_2\in C$, we have $f(a_1,b_1)=f(a_2,b_2)\Leftrightarrow b_1=b_2$.}
		\item {For all $a_1,b_1,a_2,b_2\in C$, we have $f(a_1,b_1)=f(a_2,b_2)\Leftrightarrow (a_1=a_2\text{ and }b_1=b_2)$.}
		\end{enumerate}
\end{lemma}
\begin{lemma}\label{lem:alternation-rank}
  {For {all} integers $r,c,K$ there exists an integer $N$
    such that the following property holds.  Let $G$ be a digraph with maximum
    out-degree at most $c$ and let $S, T$ be subsets of vertices of $G$, such
    that $|T|\geq N$ and for each $t,t'\in T$ there exist a vertex
    $s=s(t,t')\in S$, a directed path $P_{s,t}$ of length at most $r$ from $s$
    to~$t$ and a directed path $P_{s,t'}$ of length at most $r$ from $s$ to
    $t'$. Then $G$ contains a crown of order $ K$ as a depth-$r$
    minor.}
\end{lemma}
\begin{proof}
  {First note that we can assume that the paths from $s(t,t')$ to
    $t$ and $t'$ are non-intersecting shortest paths, that $s(t,t')$ are chosen
    in such a way that the sum of the length of the paths to $t$ and $t'$ is
    minimum, and that if the path from $s(t,t')$ to $t$ intersects the path from
    $s(t,t'')$ to $t$, then they share the same subpath after they first meet.}

  {We order $T$ as $t_1,\dots,t_n$. For every $1\leq i<j\leq n$
    we denote by $\lambda(i,j)$ the vector formed by the internal vertices of
    the path from $s(t_i,t_j)$ to $t_i$ in reverse order, followed by
    $s(t_i,t_j)$, then by the internal vertices of the path from $s(t_i,t_j)$ to
    $t_j$.  By a standard Ramsey argument, if $T$ is sufficiently large we can
    extract a large subset $T_1\subseteq T$ such that for all $ t_i,t_j \in T_1$
    with $i<j$ the path from $s(t_i,t_j)$ to $t_i$ and the path from
    $s(t_i,t_j)$ to $t_j$ have {the same} length{s} $\ell_1$, and
    $\ell_2$, respectively, independently of the choice of $i<j$. Denote by
    $\lambda_k(i,j)$, $1\leq k\leq \ell_1+\ell_2$, the $k$th component of the
    vector $\lambda(i,j)$. }

  {By applying iteratively for $k=1,2,\ldots$, \Cref{thm:CR}
    there is a large subset $T_2\subseteq T_1$ such that for every~$k$, on $T_2$
    either all vertices $\lambda_k(i,j)$ are the same or all distinct
    for all pairs $(i,j)$, or $\lambda_k(i,j)$ is an injective function of $i$
    or an injective function of $j$.  First note that no function
    $\lambda_k(i,j)$ of $(i,j)$ can be constant {if $\abs{T_2}>c^r$} as every vertex
    can reach at most $c^r$ vertices in
    $T_2$ in at most $r$ steps. Similarly, for $k\le \ell_1$, the function $\lambda_k(i,j)$ cannot
    be an injective function of $j$ and, for $k\ge \ell_1$, the function $\lambda_k(i,j)$
    cannot be an injective function of $i$. In particular, 
    $(i,j)\neq (i',j')$ implies
    $s(t_i,t_j)\neq s(t_{i'},t_{j'})$.}

    {Moreover, as we assume that if the path
    from $s(t,t')$ to $t$ intersects the path from $s(t,t'')$ to $t$, 
    then they share the same subpath after they first meet, the 
    general situation will be as follows: from the first
    coordinate to some coordinate $a<\ell_1$ the vertices are given by an
    injective function of $i$, then until some coordinate $b>\ell_1$ the
    vertices are different on all $\lambda_k(i,j)$, then after $b$ the
    vertices are given by injective functions of $j$.  Let~$t_i^-\coloneqq \lambda_a(i,j)$
    (resp. $t_i^+\coloneqq \lambda_b(i,j)$) be the element at 
    coordinate~$a$ (resp. $b$) of the vectors~$\lambda$ corresponding to $t_i$.}

  {Now observe that the assumption that paths are shortest paths
    imply that the paths from $s(t_i,t_j)$ to $t_i^-$ (resp. from $s(t_i,t_j)$
    to $t_j^+$) are vertex disjoint. However these two families may intersect,
    but each path of a family intersects at most $r$ paths from the other. By a
    standard Ramsey argument we can assume (again by considering smaller $T_3$)
    that they do not intersect. Also for $i\neq j$, the path from $t_i^-$ to $t_i$ cannot
    intersect the path from $t_j^+$ to $t_j$ as their intersection would
    contradict the minimality assumption on $s(t_i,t_j)$. Furthermore, if the
    path from $t_i^-$ to $t_i$ intersects the path from $t_i^+$ to $t_i$, then
    they share their subpath after they first meet.}

  {Contracting the paths from $t_i^-$ to $t_i$ and the paths from
    $t_i^+$ to $t_i$, as well as all remaining paths from $s(t_i,t_j)$ to
    $t_i^-$ and $t_j^+$ (excluding these vertices) we get the required crown
    shallow minor.}
\end{proof}

\begin{theorem}\label{thm:duality}
  Let $G$ be a digraph with $\dwcol{r}(G)\leq c$ and
  $S_q\not\minor_r \hspace{-2.2pt}G$. {Then there exists
  $N=N(c,q,r)$ such that 
  $\gamma_r(G)\in \Oof(\alpha_r(G)^N)$.}
\end{theorem}
\begin{proof}
  Fix an order $L$ witnessing that $\dwcol{r}(G)\leq c$.  We compute a
  distance-$r$ dominating set $D$ as follows. Initialize
  $D\coloneqq \emptyset$, $A\coloneqq \emptyset$ and
  $N\coloneqq V(G)$. While there is a vertex $v\in N$, the set of
  non-dominated vertices, pick the smallest such vertex $v$ with
  respect to $L$. Add $v$ to $A$ and $\dWReach{2r}[G,L,v]$ to
  $D$. Mark all newly dominated vertices, that is, remove
  $N_r^+[\dWReach{2r}[G,L,v]]$ from $N$. If $N=\emptyset$, return
  $D$. Clearly,~$D$ is a distance-$r$ dominating set of $G$. We prove
  that we find a large $r$-independent subset of $A$.

  Construct the undirected graph $H$ with vertex set $A$ such that two
  vertices $a,b\in A$ are connected in~$H$ if there is $u\in V(G)$
  such that $a,b\in N_r^+(u)$. An independent set in~$H$ corresponds
  to an $r$-scattered subset of $A$ in $G$.

  We claim that every vertex $u\in V(G)$ satisfies $|N_r^+(u)\cap A|\leq c$.
  Fix $u\in V(G)$. Assume towards a contradiction that $|N_r^+(u)\cap A|>
  c$. For each $a\in N_r^+(u)\cap A$ fix a path $P_{ua}$ of length at most $r$
  from $u$ to $a$. For each path $P_{ua}$, denote by $m_{ua}$ its minimal
  element with respect to $L$. Since $\dwcol{r}(G)\leq c$, we have
  $|\{m_{ua} : N_r^+(u)\cap A\}|\leq c$. Since we have more than $c$ paths
  $P_{ua}$, there must be two paths $P_{ua_1}, P_{ua_2}$, $a_1\neq a_2$, which
  have the same element $m$ as their minimal element. {Without
    loss of generality assume that $a_1<a_2$.}  Since~$m$ is the smallest
  vertex on the path $P_{ua}$, the subpath of $P_{ua_1}$ between $m$ and $a_1$
  certifies that $m$ is weakly $r$-reachable from $a_1$. Hence, when $a_1$ was
  added to $A$, the element $m$ was added to the set~$D$. Now, the subpath of
  $P_{ua_2}$ between $m$ and $a_2$ shows that $a_2$ is at distance at most $r$
  from $m$, and hence $a_2$ is marked as dominated at this point. This again
  proves $a_2\not\in A$, a contradiction.

  We now build the $r$-independence tree $T$ of $a_1,\ldots, a_\ell$ (the
  enumeration of $A$ with respect to $L$).  {Using
    \Cref{lem:alternation-rank}, we conclude that there is $N'=(c,r,q)$ such
    that $T$ does not contain a path with $s=N'$ right descendants. Let
    $N=N'+1$.}

  Hence, by \cref{lem:number-of-nodes}, if we have
  {$|A|> (m+N)^{N}$}, then we find a sequence of length $m$ with
  all left descendants. This set is an $r$-scattered set,
  which proves the theorem.
\end{proof}

Clearly, the $r$-independence tree of a sequence of vertices can be
computed in polynomial time, which gives us the following corollary.

\begin{corollary}\label{crl:dvorak}
  Let $\CCC$ be a class of digraphs which has bounded crownless
  expansion.  Then for every $r\in \N$, there is a polynomial time
  algorithm which computes a distance-$r$ dominating set $D$ with
  $|D|\leq p(\gamma_r(G))$ for some polynomial $p$.
\end{corollary}

Furthermore, we will need the duality theorem for subsets of vertices.
The following theorem is proved exactly as \cref{thm:duality},
starting the algorithm with $N=X$.

\begin{theorem}\label{thm:duality-setwise}
  Let $G$ be a digraph with $\dwcol{r}(G)\leq c$ and
  $S_q\not\minor_r G$.  Let $X\subseteq V(G)$. Denote by
  $\gamma_r(G,X)$ the size of a smallest set $D\subseteq V(G)$ with
  $X\subseteq N_r^+(D)$ and by $\alpha_r(G,X)$ the size of a largest
  set $Y\subseteq X$ such that for all $y\neq y'\in Y$ there is no
  $u\in V(G)$ with $y,y'\in N_r^+(u)$. Then {there is $N=N(c,r,q)$ such that 
  $\gamma_r(G,X)\in \Oof(\alpha_r(G,X)^N)$.}
  Furthermore, there is a polynomial time algorithm which on input
  $G$, $X\subseteq V(G)$ and $k\geq 1$ either computes a distance-$r$
  dominating set $D$ of $X$ with $|D|\leq p(k)$ for some polynomial
  $p$ {(of degree $N$)}, or outputs that no such set of size $k$ exists, together with
  an $r$-scattered set of size $k+1$.
\end{theorem}

In \cref{thm:duality-setwise}, we need to be able to compute a 
good weak $r$-coloring order $L$ in order to find the described 
distance-$r$ dominating set. We prove that this is possible in 
%the appendix (see \cref{thm:compute-wcol}). 
{the next section}. 

%\begin{theorem}[$\star$]
%  If $\CCC$ is a class of digraphs of bounded expansion, then there is
%  a function $f:\N\rightarrow \N$ and an algorithm which on
%  input $G\in \CCC$ and $r\in \N$ computes an order $L$ with
%  $|\dWReach{r}[G,L.v]|\leq f(r)$ for all $v\in V(G)$ in time $\Oof(f(r)\cdot n)$.
%\end{theorem}

\subsection{Building the kernel}\label{sec:kernel}

We prove that on classes of bounded crownless expansion we can for
every fixed value of $r$ find a polynomial kernel for the directed
distance-$r$ dominating set problem.  In the following, fix $\CCC$ and
$r$.

\begin{definition}[$r$-domination core]\label{def:domcore}
  Let~$G$ be a digraph. A set $Z\subseteq V(G)$ is an
  \emph{$r$-domination core} in~$G$ if every minimum-size set which
  $r$-dominates $Z$ also $r$-dominates~$G$.
\end{definition}

Clearly, the set $V(G)$ is an $r$-domination core. We will show how to
iteratively remove vertices from this trivial core, to arrive at
smaller and smaller domination cores, until finally, we arrive at a
core of polynomial size {in $k$}. Observe that we do not require that every
$r$-dominating set for $Z$ is also an $r$-dominating set for~$G$;
there can exist dominating sets for $Z$ which are not of minimum size
and which do not dominate the whole graph.

\begin{lemma}
  \label{lem:reduce-corce}
  There exists a polynomial $p$ and a polynomial-time algorithm that,
  given an $r$-domination core $Z \subseteq V(G)$ with $|Z| > p(k)$,
  either correctly decides that $G$ cannot be dominated by $k$
  vertices, or finds a vertex $z \in Z$ such that $Z \setminus \{z\}$
  is still an $r$-domination core.
\end{lemma}

Hence, by gradually reducing $|Z|$, we arrive at the following
theorem.

\begin{theorem}\label{thm:dcore}
  There exists a polynomial $p$ and a polynomial-time algorithm that,
  given an instance $(G,k)$ where $G\in \CCC$, either correctly
  decides that $G$ cannot be dominated by $k$ vertices, or finds an
  $r$-domination core $Z \subseteq V(G)$ with $|Z| \leq p(k)$.
\end{theorem}

\begin{proof}[of \cref{lem:reduce-corce}]
  Let $p_1$ be the polynomial of the algorithm of
  \cref{thm:duality-setwise}.  Given~$Z$, we first apply the algorithm
  of that theorem to~$G$,~$Z$, and the parameters~$r$ and~$k$.  Thus,
  we either find a distance-$r$ dominating set $Y_1$ of $Z$ such that
  $|Y_1|\leq p_1(k)$, or we find a subset $W \subseteq Z$ of size at
  least~$k+1$ that is $r$-scattered in~$G$.  In the latter case,
  since~$W$ is an obstruction to an $r$-dominating set of size at
  most~$k$, we may terminate the algorithm and provide a negative
  answer.  Hence, assume that~$Y_1$ has been successfully constructed.

  {In the first phase, w}e inductively construct sets $X_1,Y_2,X_2,Y_3,X_3,\ldots$ with
  $Y_1\subseteq X_1\subseteq Y_2\subseteq X_2\subseteq \ldots$ as
  follows:
  \begin{itemize}
  \item If $Y_i$ is already defined, then set
    $X_i=\dWReach{r}[G,L,Y_i]$.
  \item If $X_i$ is already defined, then apply the algorithm of
    \cref{thm:duality-setwise} to $G - X_i$, $Z \setminus X_i$, and
    the parameters~$r$ and $q(|X_i|)$, where {
    $q(x)=(k+1)\cdot ((r+2)\cdot c\cdot x)^c$.}
    \begin{enumerate}
    \item\label{big-scattered} Suppose the algorithm finds a set
      $W\subseteq Z \setminus X_i$ that is $r$-scattered in $G - X_i$
      and has cardinality greater than $q(|X_i|)$. Then we let
      $X = X_i$, terminate the procedure and proceed to the second
      phase with the pair $(X,W)$.
    \item\label{small-domset} Otherwise, the algorithm has found an
      $r$-dominating set $D_{i+1}$ in the graph $G-X_i$ of size at
      most $p_1(q(|X_i|))$. Then define
      $Y_{i+1}\coloneqq X_i\cup D_{i+1}$ and proceed.
    \end{enumerate}
  \end{itemize}
  Set $q'(x)=c\cdot (q(x)+p_1(x))$. As $\dwcol{r}(G)$ is bounded by
  $c$, an induction shows that $|X_i|\leq \left(q'(k)\right)^i$.
  Hence, the algorithm consecutively finds $r$-dominating sets
  $D_2,D_3,D_4,\ldots$ and constructs sets $X_2,X_3,X_4,\ldots$ up to
  the point when case (\ref{big-scattered}) is encountered.  Then the
  construction is terminated. We claim that case (\ref{big-scattered})
  is always encountered after a constant number of iterations, more
  precisely, after at most $c$ many iterations.

  Towards proving this, assume that a vertex $z$ lies in
  $Z\setminus X_i$ for some $i\geq 1$.  For each~$D_j$, which is an
  $r$-dominating set of $Z\setminus X_{j-1}$ in $G-X_{j-1}$, fix a
  shortest path $P_j$ in $G-X_{j-1}$ from some $d_j\in D_j$ to
  $z$. The smallest vertex $x_j$ on that path is weakly $r$-reachable
  from~$d_j$, as well as from $z$, and hence is added to the set $X_j$
  (see \cref{lem:wreach-sep}).  Hence, after $i$ iterations, we have
  constructed a set $\{x_1,\ldots, x_i\}$ of vertices weakly
  $r$-reachable from $z$, which proves that the procedure must stop
  after at most $\dwcol{r}(G)=c$ many steps.
  Therefore, the construction terminates within~$c$ iterations with a
  pair~$(X,W)$ with the following properties:
  \begin{itemize}
  \item $|X|\leq \left(q'(k)\right)^c$ and $|W| > q(|X|)$;
  \item $X$ $r$-dominates $Z$ (because $Y_1\subseteq X$);
  \item $W \subseteq Z \setminus X$ and $W$ is $r$-scattered in
    $G - X$.
  \end{itemize}
  We now define an equivalence relation $\simeq$ on $W$ (recall
  that $D_r^-(u,X)$ denotes the distance vector of $u$ towards $X$): 
  for $u,v\in W$, let
  \begin{align*}
    u\simeq v\quad \Leftrightarrow \quad D_r^-(u,X) = D_r^-(v,X).
  \end{align*}

  According to \cref{lem:nc}, $\simeq$ has at most
  {
    $q(x)=(k+1)\cdot ((r+2)\cdot c\cdot |X|)^c$} equivalence classes.  Since we have that
  $|W|>q(|X|)$, we infer that there is a class $\kappa$ of relation
  $\simeq$ with $|\kappa|>k+1$. Note that we can find such a class
  $\kappa$ in polynomial time, by computing the classes of $\simeq$
  directly from the definition and examining their sizes. Let~$z$ be
  an arbitrary vertex of $\kappa$. We claim that $Z\setminus \{z\}$ is
  an $r$-domination core.

  To see this, let $Z'=Z\setminus \{z\}$. Take any minimum-size set
  $D$ which $r$-dominates $Z'$ in $G$.  If $D$ also dominates $z$,
  then $D$ is a minimum-size set which $r$ dominates $Z$, hence, as
  $Z$ is an $r$-domination core, also~$D$ is an $r$-dominating set in
  $G$, and the claim follows. Hence, towards a contradiction, assume
  that $z$ is not $r$-dominated by $D$.

  Every vertex $s\in \kappa\setminus \{z\}$ is $r$-dominated by
  $D$. For each such $s$, let $v(s)$ be a vertex of $D$ that
  $r$-dominates~$s$, and let $P(s)$ be a path of length at most $r$
  that connects $v(s)$ with $s$. We claim that for each
  $s\in \kappa\setminus \{z\}$, the path $P(s)$ does not pass through
  any vertex of $X$ (in particular $v(s)\notin X$). Also,
  vertices $v(s)$ for $s\in \kappa\setminus \{z\}$ are pairwise
  distinct. Suppose otherwise and let $w$ be the vertex of
  $V(P(s))\cap X$ that is closest to $s$ on $P(s)$. Then, as
  $D_r^-(s, X)=D_r^-(z, X)$, also $z$ is $r$-dominated by $w$,
  contradicting our assumption that $z$ is not $r$-dominated by $D$.

  For the second part of the claim, suppose $v(s)=v(s')$ for some
  distinct $s,s'\in \kappa\setminus \{z\}$. Then $v(s)$ together with
  the paths $P(s)$ and $P(s')$ would contradict with the fact
  that $W$ is $r$-scattered in $G-X$.

  This however is not possible, as $\kappa$ has more than $k$
  elements.  Hence $D'$ is a $Z$-dominator, which gives us the desired
  contradiction.
\end{proof}

Now that it remains to dominate a subset $Z$, we may keep one
representative from each equivalence class in the equivalence
relation:
$u\cong_{Z,r} v\Leftrightarrow N_r^+(u)\cap Z=N_r^+(v)\cap Z$.  As
before, there are only polynomially many equivalence classes, hence
from a polynomial domination core we can construct a polynomial
kernel.

\begin{theorem}\label{thm:kernel}
  Let $\CCC$ be a class of bounded expansion. There is a polynomial
  time algorithm which on input $G$, $k$ and $r$ computes a subgraph
  $G'\subseteq G$ and a set $Z\subseteq V(G')$ such that $G$ can be
  $r$-dominated by $k$ vertices if, and only if, $Z$ can be
  $r$-dominated by $k$ vertices in $G'$ and $|Z|\leq p(k)$.
\end{theorem}

Formally, in \cref{thm:kernel} we are not computing a kernel for
distance-$r$ dominating set, as we do not compute an instance of the
distance-$r$ dominating set problem on $G'$ but rather a red-blue
instance. As observed by Drange et al.~\cite{DrangeDFKLPPRVS16}, such
an annotated instance can be translated back to the standard problem
in the following way: add two fresh vertices $w,w'$, add a directed
path of length $r$ from $w$ to $w'$, and add a directed path of length
$r$ from $w$ to each vertex of $V(G')\setminus Z$.  Then the obtained
graph $G''$ has a distance-$r$ dominating set of size at most $k+1$ if,
and only if, $G'$ admits a set of at most $k$ vertices that
$r$-dominates~$Z$.

\section{Generalized coloring numbers}

This section contains more results about the generalized coloring numbers,
which, except for \cref{thm:compute-wcol}, are not necessary for the
results presented in the main body of the paper. The results are mainly
concerned with structural properties of classes of bounded expansion 
and the weak coloring numbers. 

For many algorithmic applications it is useful to compute an order of
the vertices as a preprocessing step. For example, the coloring
number $\col(G)$ of a digraph $G$ is the minimum integer $k$ such that
there exists a linear ordering $L$ of $V(G)$ such that each vertex $v$
has at most $k$ smaller neighbors.  It is easily seen that the
coloring number of $G$ is equal to its degeneracy.  Recall that a
graph $G$ is $k$-degenerate if every subgraph of $G$ has a vertex of
degree at most $k$. Hence the coloring number is a structural measure
that measures the edge density of subgraphs of $G$. The coloring
number gets its name from the fact that it bounds the chromatic number
-- we can simply color the vertices in the order $L$ such that every
uncolored vertex gets a color not used by its at most $\col(G)$
smaller neighbors. This bound is very useful, as computing the
chromatic number of $G$ is NP-complete, whereas the coloring number
can be computed by a greedy algorithm in linear time.

%In this section we study a characterization of bounded expansion
%classes in terms of such orderings, called generalized coloring
%numbers. We will prove several structural properties of graphs in
%terms of these parameters which will be very useful for later
%algorithmic applications.

This section is structured as follows. {We already defined the weak coloring numbers $\dwcol{r}(G)$ in \Cref{sec:vc}, we will define the related measure $\dadm{r}(G)$ in
\cref{sec:def-col}.}  In \cref{sec:dtd}, we study the limit parameter
$\dwcol{\infty}(G)$. In undirected graphs, this parameter is equal to
the well known structural measure tree-depth, which motivates us to
call this new measure \emph{directed tree-depth}.  In
\cref{sec:tdcol}, we introduce the concept of low directed tree-depth
colorings, generalizing the very successful concepts of low tree-depth
colorings for undirected graphs~\cite{nevsetvril2006tree} to directed
graphs. This concept allows us to decompose a more complex graph into
a few parts whose structure is simpler and whose interaction is highly
regular. We prove that classes of directed bounded expansion are
exactly those classes which admit low directed tree-depth colorings.
Finally, in \cref{sec:tfa}, we introduce the concept of directed
transitive fraternal augmentations, which we use to compute
generalized coloring orders in linear time.

\subsection{Generalized coloring numbers}\label{sec:def-col}

For a digraph $G$ and a natural number $r$, the
\emph{$r$-admissibility} $\dadm{r}[G,L, v]$ of $v$ with respect to~$L$
is the maximum size $k$ of a family $\{P_1,\ldots,P_k\}$ of paths of
length at most~$r$ with one end~$v$, and the other end at a vertex $w$
with $w\leq_Lv$, and which satisfy $V(P_i)\cap V(P_j)=\{v\}$ for all
$1\leq i< j\leq k$. As for $r>0$ we can always let the paths end in
the first vertex smaller than $v$, we can assume that the internal
vertices of the paths are larger than~$v$.  Note that
$\dadm{r}[G,L,v]$ is an integer, whereas $\dWReach{r}[G,L, v]$ is a
vertex set.  The \emph{$r$-admissibility} $\dadm{r}(G)$ of~$G$~is
\[\dadm{r}(G)  \coloneqq  \min_{L\in\Pi(G)}\max_{v\in V(G)}\dadm{r}[G,L,v].\]

Note that $\dadm{r}(G)$ and $\dwcol{r}(G)$ are {\em
  monotone parameters}, in the sense that if $H$ is a sub-digraph of~$G$, then $\dadm{r}(H)\leq \dadm{r}(G)$ 
and $\dwcol{r}(H)\leq \dwcol{r}(G)$. As proved in
\cite{kreutzer2017structural}, for all $r\geq 1$ it holds that
$\dwcol{r}(G)\leq 2\cdot \dadm{r}(G)^r$. 

It will be very useful to
work with the following characterization of classes of bounded
expansion.

\begin{theorem}[\cite{kreutzer2017structural}]
  A class $\CCC$ of digraphs has bounded expansion if, and only if,
  there is $f\colon\N\rightarrow\N$ such that $\dwcol{r}(G) \leq f(r)$
  for all $G\in \CCC$ and all $r\geq 1$.
\end{theorem}

\subsection{The limit parameter -- directed tree-depth}\label{sec:dtd}

If $G$ is an $n$-vertex graph, we denote by $\dwcol{\infty}(G)$ the
number $\dwcol{n}(G)$. If $G$ is undirected, then
$\dwcol{\infty}(G)=\td(G)$, where $\td(G)$ denotes the tree-depth of
$G$ (see Lemma 6.5 in \cite{nevsetvril2012sparsity}). This motivates
our study of the limit parameter in directed graphs, which we call the
\emph{directed tree-depth} of a digraph~$G$.
We start we an easy example.
\begin{example}
	Let $P_{n}$ be a directed path of order $n$ (hence length $n-2$). Then 
	\[\dwcol{\infty}(P_{n})=\lceil \log_2 (n+1)\rceil.\]
        It follows from an easy induction that
        $\dwcol{\infty}(P_{n})\geq \lceil \log_2 (n+1)\rceil$: {Let $k$ be
        minimal such that}
        $n{\le} 2^k-1$ and let $r$ be the minimum vertex of $P_{n}$
        (for a linear order witnessing $\dwcol{\infty}(P_{n})$). Note
        that $r$ is weakly reachable from every vertex in $P_{n}$, and
        that the deletion of $r$ breaks $P_{n}$ into two directed path
        $P_a$ and $P_b$, one of them (say~$P_a$) having order at least
        $2^{k-1}-1$. Using the restriction of the linear order of
        $P_{n}$ on this path we get
        \[\dwcol{\infty}(P_{n})\geq \dwcol{\infty}(P_{2^{k-1}-1})+1.\]

        Conversely, ordering the vertex set of $P^{2^n-1}$ with its
        mid vertex $r$ as its minimum, then the mid vertices of the
        sub-directed paths obtained by deleting $r$, etc. we get an
        ordering of $P_{n}$ witnessing $\dwcol{\infty}(P_{n})\leq n$.
\end{example}

\begin{lemma}\label{lem:wcol_infty}
  Let $G$ be a directed graph such that $\dwcol{\infty}(G)\leq c$ for
  some constant $c$. Then $G$ does not contain a directed path with
  length greater than $2^c-2$ and all directed topological minors of
  $G$ have arc density at most~$4c$.
\end{lemma}
\begin{proof}
  If $G$ contains a directed path $P_{2^c}$ of length $2^c-1$, then
  $\dwcol{\infty}(G)\geq \dwcol{\infty}(P_{2^c})=c+1$, contradicting
  the assumption.
%
  % We prove by induction on $i\geq 0$ that a directed path $P$ of
  % length at least $2^i$ satisfies $\wcol_\infty(P)\geq i+1$. For
  % $i=0$, and a path $P$ of length at least $2^0=1$, since
  % $v\in \dWReach{\infty}[P,L,v]$ for every order~$L$ and for every
  % vertex $v$, we have $\dwcol{{\infty}}(P)\geq 1$. Assume that the
  % claim is true for some fixed value $i\geq 0$ and let~$P$ be a
  % directed path of length at least $2^{i+1}$.  Let $L$ be an order
  % of $V(P)$ such that $\max_{v\in V(G)}|\dWReach{\infty}[P,L,v]|$ is
  % minimised. Let $v$ be the smallest vertex of~$P$ with respect to
  % $L$.  Then we have $v\in \dWReach{\infty}[G,L,w]$ for all
  % $w\in V(P)$. Let~$P'$ be a directed subpath of $P-v$ of maximum
  % length.  Then $P'$ has length at least
  % $\left\lceil|P_{i+1}|/2\right\rceil-1\geq 2^{i}$.  By induction
  % hypothesis, we have $\wcol_\infty(P')\geq i$, and in particular,
  % there is a vertex $w\in V(P')$ such that
  % $|\dWReach{\infty}[P',L_{|P'},w]|\geq i$.  Hence,
  % $\dWReach{\infty}[P',L_{|P'},w]\cup \{v\}\subseteq
  % \dWReach{\infty}[P,L,w]$, which implies the claim.

  Now, let $r\coloneqq 2^c$. Then every directed topological minor
  of~$G$ is a topological depth-$r$ minor of~$G$. Let $H$ be the
  densest topological depth-$r$ minor of $G$ and assume towards a
  contradiction that $|E(H)|/|V(H)|>4c$. Let $\bar{H}$ be the
  underlying undirected graph of~$H$. We apply Proposition 1.2.2
  of~\cite{Diestel} to $\bar{H}$, that is, we iteratively remove small
  degree vertices of $\bar{H}$ to obtain a graph $\bar{H}'$ with
  minimum degree
  $\delta(\bar{H}')\geq |E(\bar{H})|/|V(\bar{H})|\geq 2c$. Let $H'$ be
  a directed subgraph of $H$ induced by $V(\bar{H}')$, where we remove
  for each pair of bi-directed arcs $(u,v), (v,u)$ one of the two arcs
  (but keep the other). Then~$H'$ has minimum degree (in-degree plus
  out-degree) at least $c$.

  Let $L$ be an order of $V(G)$ witnessing that
  $\wcol_\infty(G)\leq c$. This order also induces an order on
  $V(H')$. Let $v$ be the largest vertex of $G$ that corresponds to a
  vertex of $H'$. Then $v$ weakly reaches more than~$c$ vertices in
  $G$; it weakly reaches at least one vertex on each path connecting
  $v$ with the vertices of $G$ corresponding to its neighbors in
  $H$. Together with the vertex $v$ itself we get
  $\wcol_\infty(G)\geq \wcol_\infty(H')>c$, contradicting our
  assumption.
\end{proof}

% Motivated by these connections to the weak colouring numbers, we
% call a graph of low DAG-depth with additional density constraints a
% graph with low \emph{sparse DAG-depth}.
%
%\begin{definition}
%  Let $k\in \N$. A class $\CCC$ of digraphs has \emph{sparse
%  DAG-depth} at most $k$ if every $G\in \CCC$ does not contain a
%  directed path of length strictly longer than $k$ and all directed
%  topological minors of $G$ have arc density at most~$k$.
%\end{definition}

\subsection{Sparse directed tree-depth colourings}\label{sec:tdcol}

% Ne\v{s}et\v{r}il and Ossona de Mendez in~\cite{nevsetvril2006tree}
% introduced the concept of low tree-depth colourings (or low
% tree-depth decompositions) for undirected graphs.  This concept
% allows to decompose a more complex graph into a few parts whose
% structure is simpler and whose interaction is highly regular.  They
% proved in~\cite{nevsetvril2008grad} that also graphs which admit low
% tree-depth decompositions are uniformly sparse in a strong sense. In
% that paper they proposed the definition of bounded expansion classes
% of graphs and proved the following theorem.
%
%\begin{theorem}[Ne\v{s}et\v{r}il and Ossona de Mendez~\cite{nevsetvril2008grad}]
%  A class $\CCC$ of undirected graphs has bounded expansion if and
%  only if for every integer $p$ there is an integer $N(p)$ such that
%  every graph $G\in \CCC$ can be coloured with $N(p)$ many colours
%  such that the combination of any $i\leq p$ colour classes induces a
%  subgraph~$H$ of tree-depth at most $i$.
%\end{theorem}

In this section, we use the nice properties of the generalized
coloring numbers to decompose a more complex graph into a few parts
whose structure is simpler and whose interaction is highly
regular. More precisely, we prove that classes of directed bounded
expansion are exactly those classes which admit low directed
tree-depth colorings.

\begin{theorem}\label{thm:characterization}
  A class $\CCC$ of directed graphs has directed bounded expansion if,
  and only if, for every integer~$p$ there are integers $N(p),\ell(p)$
  and $d(p)$ such that every graph $G\in \CCC$ can be colored with
  $N(p)$ many colors such that the combination of any $i\leq p$
  color classes induces a subgraph $H$ which excludes a directed path
  of length $\ell(p)$ and all directed topological minors of $H$ have
  density at most $d(p)$.
\end{theorem}

The first direction of the theorem follows from the next lemma and
\cref{lem:wcol_infty}.

\begin{lemma}
  Let $G$ be a digraph, let $p$ be an integer and assume that
  $\dwcol{2^p}(G)\leq c$ for some constant $c$.  Then $G$ can be
  colored with $c$ colors such that the combination of any $i\leq p$
  color classes induces a subgraph~$H$ with
  $\dwcol{\infty}(H)\leq i$.
  % which excludes a directed path of length $2^i$ and all directed
  % topological minors of $H$ have density at most $4c$.
\end{lemma}
\begin{proof}
  Let $L$ be an order of $V(G)$ such that
  $|\dWReach{2^p}[G,L,v]|\leq c$ for all $v\in V(G)$.  Color the
  vertices of $G$ greedily along the order $L$ starting from the least
  element such that the color assigned to a vertex~$v$ is distinct
  from the colors assigned to $\dWReach{2^p}[G,L,v]\setminus\{v\}$.
  As $|\dWReach{2^p}[G,L,v]|\leq c$, $c$ colors suffice for this.
  
  Let $H$ be a sub-digraph induced by $i\leq p$ colors.  According to
  the coloring rule above it follows that $\dwcol{2^p}(H)\leq i$.
  
  If $H$ contains a directed path of length $2^p$ then 
  \[\dwcol{2^p}(H)\geq \dwcol{2^p}(P_{2^p})=\dwcol{\infty}(P_{2^p})>p,\]
  contradicting $\dwcol{2^p}(H)\leq i$.  It follows that
  $\dwcol{\infty}(H)=\dwcol{2^p}(H)\leq i$.
  \end{proof}

  For the other direction of~\cref{thm:characterization} we use the
  characterization of bounded expansion classes by bounded
  $r$-admissibility. For this, we need two more lemmas.  The first
  lemma describes an obstruction for small $r$-admissibility in
  directed graphs.

\begin{lemma}[see Lemma 4.7 of~\cite{kreutzer2017structural}]\label{lem:adm}
  Let $G$ be a directed graph. If $\dadm{r}(G)\geq c+1$ for some
  constant~$c$, then there exists a set $S\subseteq V(G)$ such that
  every $v\in S$ is connected to at least $c$ other vertices of $S$
  via directed paths from $v$ of length at most $r$ intersecting only
  in $v$ whose internal vertices belong to $V(G)\setminus S$.
\end{lemma}

In the above lemma, when we say that a vertex $v$ is connected to a
vertex $w$ by a directed path, we mean that $v$ and $w$ are the
end-vertices of a directed path, the path may go in either direction.

The next lemma describes the interaction of high degree vertices with
sets in bounded expansion classes.  A variant of the lemma was
originally developed for undirected graphs in~\cite{DrangeDFKLPPRVS16}
and proved in~\cite{kreutzer2017structural} for directed graphs.

Let $G$ be a digraph, $X\subseteq V(G)$, $u\in V(G)\setminus X$ and
$r\in \N$.  The \emph{$r$-projection} of $u$ onto~$X$ is the set
$M_r^G(u,X)$ of all vertices $v\in X$ such that there is a directed
path between $u$ and $v$ in~$G$ (in either direction) of length at
most $r$ with all internal vertices in $V(G)\setminus X$.

\begin{lemma}[\cite{kreutzer2017structural}]\label{lem:closure}
  Let $G$ be a digraph, $r\geq 0$ and $X\subseteq V(G)$. There exists
  a set $\mathrm{cl}_r^G(X)\subseteq V(G)$, called an
  \emph{$r$-closure of $X$ in $G$} with the following properties. Let
  $\xi:=\left\lceil 2\nabla_{r-1}(G)\right\rceil$.
  \begin{enumerate}
  \item $X\cap \mathrm{cl}_r^G(X)=\emptyset$;
  \item $|\mathrm{cl}_r^G(X)|\leq (r-1)\xi\cdot |X|$; and
  \item $|M_r^{G-\mathrm{cl}^G_r(X)}(u, X)|\leq \xi$ for all
    $u\in V(G)\setminus (X\cup \mathrm{cl}_r^G(X))$.
  \end{enumerate}
\end{lemma}

We can now prove the reverse direction of~\cref{thm:characterization}.
% A construction of a similar type was first employed
% in~\cite{grohe2015colouring}.

\begin{lemma}\label{lem:colour-be}
  Let $\CCC$ be a class of directed graphs such that for every
  integer~$r$ there are integers $N(r),\ell(r)$ and $d(r)$ such that
  every graph $G\in \CCC$ can be colored with $N(r)$ many colors
  such that the combination of any $i\leq r$ color classes induces a
  subgraph $H$ which excludes a directed path of length $\ell(r)$ and
  all directed topological minors of $H$ have density at most $d(r)$.
  Then $\CCC$ has bounded expansion.
\end{lemma}
\begin{proof}
  Let $G\in \CCC$, and let $r\geq 0$. We color $G$ with $N(r+1)$ many
  colors such that the combination of any $i\leq r+1$ color classes
  induces a subgraph $H$ which excludes a directed path of length
  $\ell(r+1)$ and all directed topological minors of $H$ have density
  at most $d(r+1)$.  According to~\cref{lem:densityminors}, all
  directed minors of $H$ have density at most
  $q\coloneqq 32\cdot (4d(r+1))^{(\ell(r+1)+1)^2}$.  We want to prove
  that $\dadm{r}(G)\leq c$ for a constant depending only on $q$ and
  $r$ to be determined in the course of the proof.

  Assume towards a contradiction that $\dadm{r}(G)\geq c+1$. According
  to~\cref{lem:adm}, there exists a set $S\subseteq V(G)$ such that
  every $v\in S$ is connected to at least $c$ other vertices of $S$
  via directed paths of length at most $r$ intersecting only in $v$
  and whose internal vertices belong to $V(G)\setminus S$. Denote the
  set of all of these paths by $\PPP$.  Since there are at most
  $x=\binom{N(r+1)}{r+1}$ possible ways to color a path of length at
  most $r$, we find a set $\PPP' \subseteq \PPP$ of paths all of which
  have the same set of colors of size $i\leq r+1$ and
  $|\PPP'|\geq (|S|\cdot c)/x$. Let~$H$ be the subgraph of~$G$ induced
  by the vertices of $\PPP'$. Let $S'\coloneqq S\cap V(H)$.  By
  definition, $H$ is colored with at most $r+1$ colors, and hence by
  assumption its directed minors have density at most~$q$.  Let
  $\xi\coloneqq \lceil 2q\rceil$.

  We construct $\mathrm{cl}_{r}^H(S')$ according
  to~\cref{lem:closure}, which has size at most
  $(r-1)\xi \cdot \abs{S'}$. We now iteratively contract short paths
  between $S'$ and $\mathrm{cl}_{r}^H(S')$.  For each path $P$
  in~$\PPP'$ with an end-vertex $v\in S'$, let $P_0$ be the
  restriction of $P$ between $v$ and the vertex
  $w\in \mathrm{cl}_{r}^H(S')\cup S'$ which is the nearest to $v$, but
  not $v$ itself.  Let $\PPP_0\coloneqq \{P_0 : P\in \PPP'\}$. If two
  paths of $\PPP_0$ have the same initial and terminal vertex (but are
  oriented in different directions), we remove one of them from
  $\PPP_0$. We hence have
  $\abs{\PPP_0}\geq \abs{S}\cdot c/(2x) \geq \abs{S'}\cdot c/(2x)$.

  Now, for $i=0,1,\ldots$, as long as there exists $P\in \PPP_i$,
  contract $P$ to an arc and remove from $\PPP_i$ all paths which
  intersect~$P$ to obtain $\PPP_{i+1}$.  We claim that every internal
  vertex~$u$ of $P$ can intersect with at most $\xi$ many other paths
  $P'\in \PPP_i$.  This is because every path $P\in \PPP'$ which uses
  vertex $u$ must have both their end-vertices in
  $M_{r}^{G-\mathrm{cl}_{r}^H(S')}(u,S')$ (by definition of $\PPP_0$
  all paths are cut when hitting $\mathrm{cl}_{r}^H(S')\cup S'$).
  Hence, as every internal vertex $u$ of $P$ satisfies
  $|M_{r}^{G-\mathrm{cl}_{r}^H(S')}(u,S')|\leq \xi$ by assumption,~$P$
  can intersect with at most $r\xi$ many other paths $P'\in \PPP_i$.

  Hence, hence after $i+1$ contractions, we have
  $|\PPP_{i+1}|\geq \frac{c}{2}|S'|-(i+1)r\xi$. Note that we are
  constructing a graph $H^*\minor_{r-1}^d H$ with vertex set
  $S'\cup \mathrm{cl}_r^H(S')$, that is, with at most
  $((r-1)\xi +1)\cdot |S'|$ vertices, which by assumption on $q$ can
  have at most $\xi/2\cdot ((r-1)\xi +1)\cdot |S'|$ many arcs.  This
  gives a contradiction for $c>r\xi^2((r-1)\xi+1)\binom{N(r+1)}{r+1}$,
  e.g. for $c=r^2\xi^3\binom{N(r+1)}{r+1}$.
\end{proof}

\subsection{Transitive fraternal augmentations}\label{sec:tfa}

To approximate the weak coloring numbers, we use the transitive
fraternal augmentation method which was employed also in the
undirected setting (see \cite{nevsetvril2008grad} and Section~7.4 of
the textbook~\cite{nevsetvril2012sparsity}).

%
% Instead, we employ the framework of transitive fraternal
% augmentations also in the directed setting.  Transitive fraternal
% augmentations for undirected graphs were introduced
% in~\cite{nevsetvril2008grad}. We refer also to Section~7.4 of the
% textbook~\cite{nevsetvril2012sparsity} for more background.

Let $G$ be a digraph. An \emph{re-orientation of $G$} is a digraph $H$
such that for each arc $(u,v)\in E(G)$ exactly one of $(u,v)$ or
$(v,u)$ is an arc of $H$. We also say that an arc set $F$ is a
re-orientation of an arc set $E$, if for each arc $(u,v)\in E$ exactly
one of $(u,v)$ or $(v,u)$ is in $F$.

For $r\in \N$ and a digraph $G$, a \emph{depth-$r$ transitive
  fraternal augmentation of $G$} is a directed graph $G_r$ with arc
set $E(G_r)$ partitioned as $E_1\cup\ldots \cup E_r$, such that
\begin{itemize}
\item the set $E_1$ is a re-orientation of $E(G)$;
\item for every every arc $(u,v)\in E_i$, $1\leq i\leq r$, there
  exists a directed path of length at most~$i$ with endpoints $u,v$
  (in either direction) in $G$;
\item $(u,v)\in E(G_r)$ implies $(v,u)\not\in E(G_r)$ for all
  $u,v\in V(G_r)$;
\item for all $1\leq i\leq j\leq r$ with $i+j\leq r$, and for all
  $u,v,w\in V(G)$, if $(w,u)\in E_i$ and $(w,v)\in E_j$ and there
  exists a directed path of length at most $i+j$ between $u$ and $v$
  in $G$, then $(u,v)$ or $(v,u)$ belongs to $\bigcup_{k=1}^{i+j}E_k$.
\item for all $1\leq i\leq j\leq r$ with $i+j\leq r$, and for all
  $u,v,w\in V(G)$, if $(u,v)\in E_i$ and $(v,w)\in E_j$ there exists a
  directed path of length at most $i+j$ between $u$ and $w$ in $G$,
  then $(u,w)$ or $(w,u)$ belongs to $\bigcup_{k=1}^{i+j}E_k$.
 \end{itemize}
 
\begin{lemma}
  Let $G$ be an $n$-vertex digraph with $\dwcol{r}(G)\leq c$. Then we
  can compute a depth-$r$ transitive fraternal augmentation $H$ with
  $\Delta^+(H)\leq 4^{r-1}(2c)^{2^{r-1}}$ in time
  $\Oof(r\cdot 4^{r-1}(2c)^{2^{r-1}}\cdot n)$.
\end{lemma}
\begin{proof}
  We compute the sets $E_1,\ldots, E_r$ as follows. We re-orient
  $E(G)$ such that the out-degree of $E_1$ is at most $c$. This is
  possible, as $\dwcol{r}(G)\leq c$ in particular implies that $G$ is
  $c$-degenerate. Hence $E_1$ satisfies the above conditions. Now, 
  {assume that 
  $E_1,\ldots, E_i$ have been constructed and satisfy the above
  conditions. For all} $1\leq j_1\leq j_2\leq i+1$ with $j_1+j_2=i+1$,
  and for all $u,v,w\in V(G)$, if $(w,u)\in E_{j_1}$ and
  $(w,v)\in E_{j_2}$ and there exists a path of length at most $i+1$
  between $u$ and $v$ in $G$, we introduce an undirected edge 
  $\{u,v\}$
  to be oriented appropriately to $E_{i+1}$. Also transitive arcs are
  introduced as undirected edges accordingly. We then orient the
  resulting edge set greedily to obtain $E_{i+1}$. Clearly, $E_{i+1}$
  again satisfies the above conditions, and hence, after $r$ steps we
  have computed a depth-$r$ transitive fraternal augmentation of
  $G$. It remains to prove the claimed degree bounds.

  We define the following arc sets $F_1,\ldots, F_r$. Fix an order $L$
  witnessing that $\dwcol{r}(G)\leq c$. The set $F_i$ contains all
  arcs $(u,v)$ such that there is a path of length at most $i$ in $G$
  between $u$ and $v$ such that $v$ is the smallest vertex of the path
  with respect $L$, {that is, the sets $F_i$ represent the weak $r$-reachability relation of~$L$}. {For each vertex $v\in V(G)$ we will determine a bound $f_i$ on the number of arcs which are oriented differently in the arc sets $E_i$ and $F_i$, more precisely, $f_i$ will be a bound on the number of arcs} with
  one end~$v$ and which are arcs of
  $\left(\bigcup_{1\leq j\leq i}E_j\setminus \bigcup_{1\leq j\leq
      i}F_j\right) \cup \left(\bigcup_{1\leq j\leq i}F_j\setminus
    \bigcup_{1\leq j\leq i}E_j\right)$.
%  That is, at every vertex the two arc sets differ by at most $f_i$
%  orientations.

  First, let $i=1$. Since $v$ has at most $c$ smaller neighbors in
  the first greedy orientation, at most $c$ arcs of~$E_1$ may be
  directed away from $v$, which are all directed towards $v$ in
  $F_1$. On the other hand, there may be $c$ arcs in $F_1$ being
  directed towards~$v$ which are directed away from $v$ in $E_1$. In
  total, we have at most $2c$ wrongly directed arcs and we define
  $f_1=2c$.

  Now assume that we have defined the number $f_i$ for some fixed
  $i\geq 1$. Consider $v\in V(G)$ and see how many undirected edges
  including $v$ are introduced to $E_{i+1}$, which are not also edges
  of $F_{i+1}$. First observe that for each triple $u,v,w$, if we have
  only one wrongly directed arc, then we do not introduce an edge
  which is not also present in $F_{i+1}$. Consider e.g.\ the case that
  there is an arc $(w,v)\in F_{j_1}\cap E_{j_1}$ and an arc
  $(u,w)\in E_{j_2}$ with $(w,u)\in F_{j_2}$, $j_1+j_2=i+1$. Then we
  have $(v,u)\in F_{i+1}$ as a transitive arc, while we have $\{u,v\}$
  as a fraternal edge to be directed in $E_{i+1}$. The other cases are
  similar.

  Hence a wrongly oriented edge $\{u,v\}$ can only be introduced if
  there is a vertex $w\in V(G)$ such that both $(v,w)$ and $(u,w)$ are
  wrongly oriented. However, there are at most $f_i$ such choices for
  $w$ and each such $w$ also has at most so many bad choices. Hence,
  every vertex has at most $f_i^2$ wrongly oriented edges in $E_{i+1}$
  which are not edges of $F_{i+1}$. Hence, as $F_{i+1}$ is
  $c$-degenerate, $E_{i+1}$ is $c+f_i^2$-degenerate and the greedy
  orientation procedure will produce an orientation which for every
  vertex coincides on all but $2(c+f_i^2)\leq 2(2f_i^2)=4f_i^2$ edges,
  as in the case $i=1$.  We can hence define
  $f_i\coloneqq 4^{i-1}(2c)^{2^{i-1}}$ and conclude.

  For the running time, observe that a greedy orientation of a graph
  with $m$ edges can be computed in time $\Oof(m)$. As we have to
  compute $r$ orientations on graphs with at most
  $4^{r-1}(2c)^{2^{r-1}}\cdot n$ edges, the claim follows.
\end{proof}

We now show that transitive fraternal augmentations can be used to
compute good linear orders for the weak coloring numbers. We need one
more lemma.

\begin{lemma}\label{lem:path-aug}
  Let $P$ be a directed path of length at most $r$ in a digraph $G$
  with end vertices $u,v$. Then in every depth-$r$ transitive
  fraternal augmentation $H$ of $G$, either $(u,v)$ or $(v,u)$ are
  arcs of $H$, or there is $w\in V(G)$ such that $(u,w)$ and $(v,w)$
  are arcs of $H$.
\end{lemma}
\begin{proof}
  Since the appropriate sub-paths of $P$ witness the existence of paths
  of length~$i$ in the $i$-th augmentation step, we can argue as in
  the undirected case, compare to Lemma 7.9
  of~\cite{nevsetvril2012sparsity}.
\end{proof}

\begin{lemma}\label{lem:greedy-orientation}
  Let $L$ be a greedy orientation of a depth-$r$ transitive fraternal
  augmentation~$H$ of $G$ with $\Delta^+(H)\leq d$, such that every
  vertex has at most $c$ smaller neighbors with respect to $L$. Then
  \[|\dWReach{r}[G,L,v]|\leq (d+1)c+1\] for all $v\in V(G)$.
\end{lemma}
\begin{proof}
  For each vertex $v\in V(G)$ we count the number of end-vertices of
  paths of length at most $r$ from $v$ such that the end-vertex is the
  smallest vertex of the path. This number is exactly
  $|\dWReach{r}[G, L, v)]|$.

  By \cref{lem:path-aug}, for each such path with end-vertex
  $w \neq v$, we either have an arc $(v, w)$ or an arc $(w, v)$ in $H$
  or there is $u$ on the path and we have arcs $(v, u)$, $(w, u)$ in
  $H$.  By assumption on $L$ there are at most $c$ arcs $(v, w)$ or
  $(w, v)$ such that $w <_L v$. Furthermore, we have at most $d$ arcs
  $(v, u)$, as~$v$ has out-degree at most $d$ and for each such $u$
  there are at most $c$ arcs $(w, u)$ such that $w <_L u$ by
  assumption on~$L$. These are exactly the pairs of arcs we have to
  consider, as no vertex on the path from~$v$ to $w$ may be smaller
  than~$w$. Hence in total we have
  $|\dWReach{r}[G,L,v]|\leq c+d\cdot c+1 =(d+1)c+1$.
\end{proof}

Algorithmically, to obtain a good order from
\cref{lem:greedy-orientation}, we have to compute one final greedy
orientation step. We hence obtain the following theorem.

\begin{theorem}\label{thm:compute-wcol}
  If $\CCC$ is a class of digraphs of bounded expansion, then there is
  a function $f:\N\rightarrow \N$ and an algorithm which on
  input $G\in \CCC$ and $r\in \N$ computes an order $L$ with
  $|\dWReach{r}[G,L.v]|\leq f(r)$ for all $v\in V(G)$ in time $\Oof(f(r)\cdot n)$.
\end{theorem}

%%% Local Variables:
%%% mode: latex
%%% TeX-master: "soda"
%%% End:

\bibliographystyle{abbrv}

\end{document}